\documentclass{lmcs} %
\pdfoutput=1
\usepackage[utf8]{inputenc}

\usepackage{lastpage}
\lmcsdoi{22}{1}{15}
\lmcsheading{}{\pageref{LastPage}}{}{}%
{Dec.~23,~2024}{Feb.~27,~2026}{}


\usepackage{hyperref}

\usepackage{xspace}
\usepackage{latexsym}
\usepackage{amsmath,amstext,amsfonts,amssymb,amsthm}
\usepackage{wasysym}

\usepackage{quiver}

\usepackage{graphicx}

\usepackage{stmaryrd}

\usepackage{tikz}
\usetikzlibrary{arrows,arrows.meta,automata,backgrounds,calc,chains,decorations.fractals,decorations.pathreplacing,fadings,fit,folding,mindmap,patterns,plotmarks,positioning,shadows,shapes.geometric,shapes.symbols,through,trees,cd}

\colorlet{dblue}{blue!40!black}

\usepackage{textcomp}

\usepackage{enumitem}
\setlist{leftmargin=1cm}

\usepackage{verbatim}
\usepackage{float}
\usepackage{pdfpages}

\usepackage{wrapfig}

\usepackage{fontawesome}

\usepackage{doi}
\usepackage{orcidlink}
\usepackage{xcolor,hyperref}

\usepackage{xspace}
\usepackage{mathtools}

\usepackage{caption,subcaption}
\DeclareCaptionSubType[alph]{figure}
\usepackage{adjustbox}
\usepackage{cleveref}
\usepackage{stackrel}

\usepackage{circledsteps}

\newcommand{\card}[1]{\##1}

\newcommand{\pbpostrong}{PBPO$^{+}$\xspace}
\newcommand{\catname}[1]{{\normalfont\textbf{#1}}\xspace}
\newcommand{\Set}{\catname{Set}}

\newcommand{\Graph}{\catname{Graph}}
\newcommand{\SimpleGraph}{\catname{SGraph}}

\newcommand{\CC}{\catname{C}}

\let\emptyset\varnothing

\newcommand{\mono}{\rightarrowtail}

\newcommand{\epi}{\twoheadrightarrow}

\newcommand{\iso}{\cong}

\newcommand{\domain}[1]{\mathit{dom}(#1)}

\tikzset{default/.style={
  thick,
  every node/.style={circle},
  level distance=12mm, 
  inner sep=.5mm}}
\tikzset{smallCircle/.style={circle,fill=black,inner sep=0mm,outer sep=1mm,minimum size=1mm}}
\tikzset{paint/.style={very thick,draw=#1!50!black,fill=#1,opacity=.4}}
\tikzset{paintopaque/.style={very thick,draw=#1!50!black!60,fill=#1!60}}
\tikzset{loop/.style={out=-30+#1,in=30+#1,distance=2.7em,pos=0.5}}
\tikzset{thinloop/.style={out=-20+#1,in=20+#1,distance=2.7em,pos=0.5}}

\tikzset{emptystep/.style={-,dotted,line cap=round,dash pattern=on 0 off 3.00000}}

\tikzset{b/.style={anchor=north,at=(#1.south)}}
\tikzset{br/.style={anchor=north west,at=(#1.south east)}}
\tikzset{bl/.style={anchor=north east,at=(#1.south west)}}
\tikzset{bw/.style={anchor=north west,at=(#1.south west)}}
\tikzset{be/.style={anchor=north east,at=(#1.south east)}}
\tikzset{a/.style={anchor=south,at=(#1.north)}}
\tikzset{ar/.style={anchor=south west,at=(#1.north east)}}
\tikzset{al/.style={anchor=south east,at=(#1.north west)}}
\tikzset{aw/.style={anchor=south west,at=(#1.north west)}}
\tikzset{ae/.style={anchor=south east,at=(#1.north east)}}
\tikzset{r/.style={anchor=west,at=(#1.east)}}
\tikzset{l/.style={anchor=east,at=(#1.west)}}
\tikzset{rn/.style={anchor=north west,at=(#1.north east)}}

\tikzset{eta/.style={very thick,->,cblue!90!black}}
\tikzset{beta/.style={very thick,->,corange!80!white!90!black}}

\tikzset{devcirc/.style={circle,draw,fill=white,inner sep=0,minimum size=4\pgflinewidth}}
\tikzset{dev/.style={postaction={decorate},decoration={
  markings,
  mark=at position .5 with \node [devcirc] {};}}
}
    
\tikzset{medium tree/.style={
    level 1/.style={sibling distance=17mm},
    level 2/.style={sibling distance=9mm},
    level 3/.style={sibling distance=5mm},
    level 4/.style={sibling distance=4mm},
  }}

\tikzset{startAt/.style={inner sep=0mm,r=#1,xshift=-1.2mm,yshift=.8mm}}

\tikzset{graphNode/.style={circle,draw=black,inner sep=.5mm,outer sep=.5mm}}
\tikzset{morph/.style={-{>[width=1.7mm]}}}
\tikzset{mono/.style={{>[width=1.7mm]}-{>[width=1.7mm]}}}
\tikzset{regmono/.style={right hook-{>[width=1.7mm]}}}
\tikzset{edge/.style={line width=0.2mm,-{Triangle[width=1.4mm]}}}

\newcommand{\arrowTriangle}[2]{
  \pgfpathmoveto{\pgfpoint{-0.01#1+#2}{.6#1}}
  \pgfpathlineto{\pgfpoint{#1+#2}{0}}
  \pgfpathlineto{\pgfpoint{-0.01#1+#2}{-.6#1}}
  \pgfusepathqfill
}

\newdimen\prearrowsize
\newdimen\arrowsize
\newdimen\temparrowsize
\newcommand{\arrowscale}{5}
\newcommand{\setarrowsize}{
  \arrowsize=0.000000001pt
  \prearrowsize=\arrowscale\pgflinewidth
  \normalizearrowsize
}
\newcommand{\normalizearrowsize}{
  \ifdim\prearrowsize>2mm
    \addtolength{\arrowsize}{2mm}
    \addtolength{\prearrowsize}{-2mm}
    \temparrowsize=0.5\prearrowsize
    \prearrowsize=\temparrowsize
  \else
  \fi

  \addtolength{\arrowsize}{\prearrowsize}
}

\pgfarrowsdeclarealias{share>}{share>}{stealth'}{stealth'}%

\pgfarrowsdeclare{red>}{red>}
{
  \setarrowsize
  \pgfarrowsleftextend{1.5\pgflinewidth} 
  \pgfarrowsrightextend{\arrowsize-0.1\arrowsize}
}{
  \setarrowsize
  \pgfsetdash{}{0pt} %
  \pgfsetroundjoin %
  \pgfsetroundcap %
  \arrowTriangle{\arrowsize}{-0.1\arrowsize}
}

\pgfarrowsdeclare{red>>}{red>>}
{
  \setarrowsize
  \pgfarrowsleftextend{1.5\pgflinewidth}
  \pgfarrowsrightextend{\arrowsize-0.1\arrowsize+.8\arrowsize}
}{
  \setarrowsize
  \pgfsetdash{}{0pt} %
  \pgfsetroundjoin %
  \pgfsetroundcap %
  \arrowTriangle{\arrowsize}{-0.1\arrowsize}
  \arrowTriangle{\arrowsize}{-0.1\arrowsize+.8\arrowsize}
}

\pgfarrowsdeclare{red>>>}{red>>>}
{
  \setarrowsize
  \pgfarrowsleftextend{1.5\pgflinewidth}
  \pgfarrowsrightextend{\arrowsize-0.1\arrowsize+1.6\arrowsize}
}{
  \setarrowsize
  \pgfsetdash{}{0pt} %
  \pgfsetroundjoin %
  \pgfsetroundcap %
  \arrowTriangle{\arrowsize}{-0.1\arrowsize}
  \arrowTriangle{\arrowsize}{-0.1\arrowsize+.8\arrowsize}
  \arrowTriangle{\arrowsize}{-0.1\arrowsize+1.6\arrowsize}
}

\tikzstyle{gyellow}=[draw=black!80,top color=white!50,bottom color=black!20]
\tikzstyle{gblue}=[draw=blue!50,top color=white,bottom color=blue!60]
\tikzstyle{gred}=[draw=red!50,top color=white,bottom color=red!60]
\tikzstyle{ggreen}=[draw=blue!80!green!90!black,top color=white,bottom color=blue!80!green!60]

\tikzstyle{roundNode}=[gyellow,thick,circle,minimum size=4mm,inner sep=0.5mm]

\definecolor{cblue}{rgb}{0,0.4,0.7}
\definecolor{clighterblue}{rgb}{0,0.6,1.0}
\colorlet{cred}{red}
\colorlet{cgreen}{green!80!black}
\colorlet{corange}{orange!70!red}
\colorlet{cpureorange}{orange}
\colorlet{cpurple}{clighterblue!50!cred}

\colorlet{clightblue}{clighterblue!50!cblue!40}%
\colorlet{clightred}{cred!40}
\colorlet{clightgreen}{cgreen!80!cblue!40}
\colorlet{clightyellow}{corange!40!yellow!50}
\colorlet{clightorange}{cred!50!orange!40}
\colorlet{clightpurple}{clighterblue!50!cred!50}

\colorlet{cdarkred}{cred!70!black}
\colorlet{cdarkgreen}{cgreen!60!black}
\colorlet{cdarkblue}{cblue!60!black}

\colorlet{chighlight}{orange!50!yellow!60}

\tikzset{pgnode/.style={smallCircle,fill=white,draw=black,minimum size=1.3mm,outer sep=0.5mm}}
\tikzset{pgnodecolor/.style={pgnode,minimum size=4mm,scale=0.9}}
\tikzset{pgnodebig/.style={roundNode,gyellow,outer sep=1mm}}
\tikzset{pgrelation/.style={ultra thick,cblue!80!black,decorate,decoration={snake,amplitude=.4mm,segment length=4mm}}}

\tikzset{exi/.style={densely dotted}}
\tikzset{decweak/.style={-,green!50!black,very thick,opacity=0.7}}
\tikzset{decstrict/.style={->,orange,very thick,opacity=0.7}}

\tikzset{graphNode/.style={circle,draw=black,inner sep=.5mm,outer sep=1mm}}

\tikzset{sloop/.style={looseness=7}}

\tikzset{interconnect/.style={dotted,cred}}
\tikzset{match/.style={cdarkgreen,very thick}}
\tikzset{jigsaw/.style={circle,draw=black,minimum size=2mm,fill=white,minimum size=3.5mm}}

\colorlet{myblue}{blue!80!black}
\colorlet{mygreen}{cdarkgreen}
\colorlet{myred}{cred}
\colorlet{myorange}{corange}
\colorlet{mypurple}{blue!40!cred}

\tikzset{smalljigsaw/.style={rectangle,rounded corners=3mm,inner sep=1mm}}

\tikzset{epattern/.style={}}
\tikzset{eset/.style={draw=black!50}}
\tikzset{npattern/.style={rectangle,rounded corners=2mm,draw=black,inner sep=0.4mm,outer sep=.5mm,minimum size=4.5mm}}
\tikzset{nset/.style={npattern,draw=black!50,fill=white}}
\tikzset{label/.style={rectangle,scale=0.85,inner sep=0.5mm,outer sep=0.5mm,fill=white}}
\tikzset{tlabel/.style={rectangle,scale=0.85,inner sep=0.5mm,outer sep=0.5mm}}

\tikzset{short/.style={node distance=10mm}}

\newcommand{\annotate}[2]{
  \node at (#1.north east) [anchor=south west,xshift=-.75mm,yshift=-.75mm,label,outer sep=0mm] {#2};
}

\newcommand{\graphbox}[8]{
  \begin{scope}[xshift=#2,yshift=#3]
    \draw [rounded corners=2mm] (0,0) rectangle (#4,-#5);
    \node at (0,0mm) [anchor=north west,inner sep=1mm] {#1};
    \begin{scope}[xshift=#4/2+#6,yshift=#7] 
    #8
    \end{scope}
  \end{scope}
}

\newcommand{\vertex}[2]{%
  \begin{tikzpicture}[baseline=-1ex]%
    \node [rectangle,rounded corners=1.5mm,inner sep=0mm,minimum size=3.5mm,fill=#2] {$#1$};%
  \end{tikzpicture}%
}

\tikzset{styleNode/.style={graphNode,rectangle,rounded corners=2mm,fill=white,inner sep=1.5mm,outer sep=0}}

\tikzset{styleEdge/.style={->,shorten <= 2mm,shorten >= 2mm}}
\tikzset{styleEdgeNode/.style={rectangle,rounded corners=2mm,fill=white,draw=none,outer sep=0,inner sep=1mm,pos=0.45}}

\tikzcdset{scale cd/.style={every label/.append style={scale=#1},
    cells={nodes={scale=#1}}}}

\DeclarePairedDelimiterX\setvbar[2]{\{}{\}}{#1 \nonscript\;\delimsize \vert \nonscript\; #2}
\DeclarePairedDelimiterX\setcolon[2]{\{}{\}}{#1 : #2}

\renewcommand{\geq}{\geqslant}

\makeatletter
\newcommand{\doublearrowtail@inner}[2]{%
  \vcenter{\offinterlineskip
    \halign{%
      ##\cr
      $\m@th#1\rightarrowtail$\cr
      \makebox[\widthof{$\m@th#1\rightarrowtail$}][s]{%
        $\m@th#1\leftarrowtail$%
      }\cr
    }%
  }%
}

\DeclareRobustCommand{\doublearrowtail}{%
  \mathrel{\mathpalette\doublearrowtail@inner\relax}%
}
\makeatother

\newcommand{\bigsub}[1]{{\thickmuskip=2mu plus 1mu minus 1mu\scriptsize\begin{array}{@{}c@{}}#1\end{array}}}

\newcommand{\elements}{\mathbb{E}}
\newcommand{\awtg}{\mathcal{T}}
\newcommand{\wtg}{(T, \elements, S, \ww)}

\newcommand{\DPOstep}[4]{{#1} \Rightarrow_\mathrm{DPO}^{#3, #4} {#2}}

\newcommand{\leftto}{\leftarrow}

\tikzset{graphNode/.style={circle,draw=black,inner sep=.5mm,outer sep=1mm}}

\renewcommand{\emptyset}{\varnothing}

\newcommand{\set}[1]{\{\,#1\,\}}
\newcommand{\tuple}[1]{\langle\,#1\,\rangle}

\newcommand{\categoryc}{\catname{C}}

\newcommand{\Hom}{\mathrm{Hom}}
\newcommand{\homset}[2]{\Hom(#1, #2)}
\newcommand{\ob}[1]{\mathit{Ob}(#1)}

\newcommand{\compi}[2]{#2 \circ #1}

\newcommand{\counting}[2]{\#\set{ #2 \circ {-} = #1 }}

\newcommand{\countingmin}[4]{
    \# \set{ \tau \in \set{ #2 \circ {-} = #1 } \mid \not\exists #4.\ \tau = {#3 \circ #4} }
}

\newcommand{\ww}{\mathbf{w}}
\newcommand{\weight}[2][]{\ww_{#1}(#2)}
\newcommand{\weighti}[3][]{\ww_{#1}(\set{{-} \circ #3 = #2})}

\newcommand{\step}[2][]{\Rightarrow^{#1}_{#2}}

\newcommand{\arule}{\rho}
\newcommand{\aframework}{\mathfrak{F}}
\newcommand{\asystem}{\mathfrak{S}}
\newcommand{\leftsq}[1]{\textit{left}(#1)}
\newcommand{\rightsq}[1]{\textit{right}(#1)}
\newcommand{\posquare}[8]{%
  \langle\, 
  #5 \stackrel{#4}{\leftarrow} #1 \stackrel{#2}{\to} #3;
  \;
  #5 \stackrel{#6}{\to} #8 \stackrel{#7}{\leftarrow} #3
  \,\rangle
}

\newcommand{\ozero}{\mathbf{0}}
\newcommand{\oone}{\mathbf{1}}
\newcommand{\ostrict}{\prec}
\newcommand{\ostricti}{\succ}
\newcommand{\oweak}{\le}
\newcommand{\oweaki}{\ge}
\newcommand{\wf}{well-founded}
\newcommand{\SN}{\mathrm{SN}}
\newcommand{\ring}{\tuple{S,\oplus,\otimes,\ozero,\oone}}
\newcommand{\wfring}{\tuple{S,\oplus,\otimes,\ozero,\oone,\ostrict,\oweak}}

\newcommand{\edgeGraph}[1][]{
  \,\begin{tikzpicture}[baseline=-.6ex]
    \node (x) [circle,inner sep=0,outer sep=1mm,minimum size=0.7mm,fill=black] {};
    \node (y) [circle,inner sep=0,outer sep=1mm,minimum size=0.7mm,fill=black, right of=x, xshift=-3mm] {};
    \draw [->] (x) to node [above] {#1} (y);
  \end{tikzpicture}\,
}

\newcommand{\loopGraph}{
  \,\begin{tikzpicture}[baseline=-.6ex]
    \node (x) [circle,inner sep=0,outer sep=1mm,minimum size=0.7mm,fill=black] {};
    \draw [->] (x) to[loop=0,distance=1em] (x);
  \end{tikzpicture}\,
}

\newcommand{\nodeGraph}{
  \,\begin{tikzpicture}[baseline=-.6ex]
    \node (x) [circle,inner sep=0,outer sep=1mm,minimum size=0.7mm,fill=black] {};
  \end{tikzpicture}\,
}

\newcommand{\nat}{\mathbb{N}}
\newcommand{\osum}{\bigoplus}
\newcommand{\oprod}{\bigodot}
\renewcommand{\otimes}{\odot}

\newcommand{\flower}{\mathfrak{c}}

\tikzset{xmonic/.style={Turned Square[open]->}}

\theoremstyle{plain}

\newtheorem{notation}[thm]{Notation}
\newtheorem{open}[thm]{Open Question}

\newcommand{\envalias}[2]{\newenvironment{#1}{\begin{#2}}{\end{#2}}}
\envalias{theorem}{thm}
\envalias{corollary}{cor}
\envalias{lemma}{lem}
\envalias{proposition}{prop}
\envalias{remark}{rem}
\envalias{example}{exa}
\envalias{definition}{defi}
\envalias{conjecture}{conj}

\colorlet{linkcolor}{red!60!black}
\hypersetup{
    colorlinks=true,
    linkcolor=linkcolor,
    citecolor=linkcolor,
    filecolor=linkcolor,      
    urlcolor=linkcolor,
    pdftitle={Generalized Weighted Type Graphs for Termination of Graph Transformation Systems},
    }

\begin{document}

\title[Generalized Weighted Type Graphs for Termination]{Termination of Graph Transformation Systems\\via Generalized Weighted Type Graphs}

\author{J\"{o}rg Endrullis\orcidlink{0000-0002-2554-8270}}
\author{Roy Overbeek\orcidlink{0000-0003-0569-0947}}

\address{Vrije Universiteit Amsterdam, Amsterdam, The Netherlands}
\email{j.endrullis@vu.nl, overbeek.research@outlook.com} 

\let\oldemph\emph
\renewcommand{\emph}[1]{{\normalfont\textbf{#1}}}

\begin{abstract}
   We refine the weighted type graph technique for proving termination of double pushout (DPO) graph transformation systems.
   We increase the power of the approach for graphs,
   we generalize the technique to other categories, and we allow for variations of DPO that occur in the literature.
\end{abstract}

\maketitle

\section{Introduction}
\label{sec:introduction}

Termination is a central aspect of program correctness.
There has been extensive research on termination~\cite{SVCOMP22,lee2001,cook2006,cook2011,apt1993logic,aproveHaskell} of imperative, functional, and logic programs. Moreover, there exists a wealth
of powerful termination techniques for term rewriting systems~\cite{terese:2003:zantema}. 
Term rewriting has proven useful as an intermediary formalism to reason about C, Java, Haskell, and Prolog programs~\cite{competition,aprove}.

There are limitations to the adequacy of terms, however. Terms are usually crude representations of program states, and termination might be lost in translation. 
This holds especially for programs that operate with graph structures, and for programs that make use of non-trivial pointer or reference structures.

Nonetheless, such programs can be elegantly represented by graphs. This motivates the use of graph transformation systems as a unifying formalism for reasoning about programs. Correspondingly, there is a need for automatable techniques for reasoning about  graph transformation systems, in particular for (dis)proving properties such as termination and confluence. %

Unfortunately, there are not many techniques for proving termination of graph transformation systems. A major obstacle is the large variety of graph-like structures. 
Indeed, the termination techniques that do exist are usually defined for rather specific notions of graphs. 
This is in stark contrast to the general philosophy of algebraic graph transformation~\cite{ehrig2006fundamentals}, the predominant approach to graph transformation, which uses the language of category theory to define and study graph transformations in a graph-agnostic manner.

In this paper, we propose a powerful graph-agnostic technique for proving termination of graph transformation systems based on weighted type graphs. We base ourselves on work by
Bruggink et al.~\cite{bruggink2015}.
Their method was developed specifically for edge-labelled multigraphs, using double pushout (DPO) graph transformation, and the graph rewrite rules were required to have a discrete interface, i.e., an interface without edges. We generalize and
improve
their approach in several  ways:

\begin{enumerate}[label=(\alph*)]
\item
    We strengthen the weighted type graphs technique for graph rewriting with monic matching. 
    In the DPO literature, matches are usually required to be monic since this increases expressivity~\cite{habel2001doublerevisited}.
    The method of Bruggink et al.~\cite{bruggink2015} is not sensitive to such constraints. It always proves the stronger property of termination with respect to unrestricted matching. 
    Consequently, it is bound to fail whenever termination depends on monic matching.
    Our method, by contrast, is sensitive to such constraints.

\item
    We generalize weighted type graphs to arbitrary categories, and we propose the notion of `traceable (sub)objects' as a generalization of nodes and edges to arbitrary categories.

\item
    We formulate our technique on an abstract level, by describing minimal assumptions on the behaviour of pushouts involved in the rewrite steps. This makes our technique applicable to the many variants of DPO, such as those that restrict to specific pushout complements (see Remark~\ref{remark:dpo:variants}).

\end{enumerate}

\subsection*{Outline}

In Section~\ref{sec:overview}, we give a gentle introduction to our termination technique. 
Section~\ref{sec:preliminaries} contains preliminaries on rewriting, termination, and semirings.
We introduce a notion of weighted type graphs for general categories in Section~\ref{sec:typegraphs}.
We then investigate traceability and monicity criteria that enable us to decompose the weight of pushout objects in~Section~\ref{sec:weighing}.
In Section~\ref{sec:termination}, we propose a technique for proving termination of graph transformation systems.
We demonstrate our technique on a number of examples in Section~\ref{sec:examples}. 
In Section~\ref{sec:benchmarks}, we provide some benchmarks from a Scala implementation.
We conclude with a discussion of related work in Section~\ref{sec:related}, followed by directions for future research in Section~\ref{sec:future}.

\section{Overview of the Technique}
\label{sec:overview}

\newcommand{\basestepleft}{
  \node (I) {$K$};
  \node (L) [left of=I,leftsq] {$L$};
  \node (G) [below of=L,leftsq] {$G$};
  \node (C) [below of=I] {$C$};
  
  \draw [->,leftsq] (I) to node [label] {$l$} (L);
  \draw [->,leftsq] (L) to node [label,pos=0.4] {$m$} (G);
  \draw [->] (I) to node [label,pos=0.4] {$u$} (C);
  \draw [->,leftsq] (C) to node [label] {$l'$} (G);

  \node [at=($(I)!.5!(G)$),leftsq] {\normalfont PO};
}
\newcommand{\basestep}{
  \basestepleft
  \node (R) [right of=I,rightsq] {$R$};
  \node (H) [below of=R,rightsq] {$H$};
  
  \draw [->,rightsq] (I) to node [label] {$r$} (R);
  \draw [->,rightsq] (R) to node [label,pos=0.4] {$w$} (H);
  \draw [->,rightsq] (C) to node [label] {$r'$} (H);

  \node [at=($(I)!.5!(H)$),rightsq] {\normalfont PO};
}

In this section, we give a gentle introduction to our termination technique and explain the key properties that we employ to weigh pushouts.

The general idea behind our technique is as follows.
DPO rewrite systems induce rewrite steps $G \Rightarrow H$ on objects if there exists a diagram of the form
\begin{equation}
        \begin{tikzpicture}[node distance=14mm,leftsq/.style={},rightsq/.style={},baseline=-7mm]
            \basestep
        \end{tikzpicture}
        \label{rewritestep}
\end{equation}
The top span of the diagram is a rule of the rewrite system. We define a decreasing measure~$\ww$ on objects, and show that $\weight{G} > \weight{H}$ for any such step generated by the system. We want to do this in a general categorical setting. This means that we cannot make assumptions about what the objects and morphisms represent. For simplicity, assume for now that the system contains a single rule $\rho: L \leftto K \to R$.

A preliminary measure $\ww$ is $\ww = \card{\homset{\,\cdot\,}{T}}$ for some distinguished object $T$. Given such a $T$, a termination proof then consists of finding bounds $B_L$ and $B_R$ such that for all steps $G \Rightarrow H$ generated by $\rho$, 
\begin{align*}
    \card{\homset{G}{T}} \geq B_L > B_R \geq \card{\homset{H}{T}} \;.  
\end{align*}
These bounds would have to be constructed using only $\rho$, $T$, and generic facts and assumptions about pushout squares. The fact that the two pushout squares of a DPO step share the vertical morphism $u: K \to C$ is useful here, because naively, $G = {L \cup_K C}$ and $H = {R \cup_K C}$, meaning that $K$ and $C$ are factors common to $G$ and $H$.

In the remainder of this section, we refine the preliminary measure into something more powerful (Section~\ref{sec:weighing:morphisms:into:a:type:graph}), explain which bounds we want to consider (Section~\ref{sec:bounds}), and explain how to construct these bounds from pushout properties (Section~\ref{sec:assumptions:about:pushout:squares}).

\subsection{Weighing morphisms into a type graph}
\label{sec:weighing:morphisms:into:a:type:graph}
We equip $T$ with a \emph{weighted set of $T$-valued elements $(\elements, \ww)$}, where $\elements$ contains morphisms $e$ with codomain $T$, and $w : \elements \to \mathcal{S}$ is a weight assignment into a well-founded commutative semiring $\mathcal{S}$. This gives rise to a notion of \emph{weighted type graphs} (see Definition~\ref{def:wtg}).

Morphisms $\phi : G \to T$ are assigned a weight depending on how they relate to these weighted elements. 
A priori there are different ways of doing this, but we will count for each $e \in \elements$ the number $n(e)$ of commuting triangles of the form
\begin{equation}
    \label{eq:commuting:triangle}
    \hspace{-1mm}
    \begin{tikzpicture}[baseline=-1.5ex,node distance=10mm,nodes={rectangle,inner sep=1mm,outer sep=0}]
        \node at (0,-1mm) {$=$};
        \node (G) at ($(30:7mm) + (3mm,0)$) {$G$};
        \node (T) at (-90:8mm) {$T$};
        \node (Ti) at ($(150:7mm) + (-3mm,0)$) {$\domain{e}$};
        \draw [->] (G) to node [label,pos=0.4] {$\phi$} (T);
        \draw [->] (Ti) to node [label,pos=0.4] {$\alpha$} (G); 
        \draw [->] (Ti) to node [label,pos=0.4] {$e$} (T); 
    \end{tikzpicture}
\end{equation}
We then compute the \emph{weight of a morphism $\phi : G \to T$} by
\begin{align}
    \ww(\phi) = \prod_{e \in \elements} \ww(e)^{n(e)} 
\end{align}
Finally, the \emph{weight of an object $G$} is defined as the sum of the weights of the  morphisms in $\homset{G}{T}$ (see further Definition~\ref{def:weighing}). 

\subsection{Lower and upper bounds for pushouts}
\label{sec:bounds}
Our goal is to establish that the rewrite step is decreasing, that is, $\weight{G} > \weight{H}$, by considering only the rule $\rho: L \leftto K \to R$.
For this purpose, we aim to decompose the weights of $G$ and $H$ as follows:
    \begin{align}
      \weight{G} \;\ge\; \weight{C} \cdot \weight{L} 
      &&
      \weight{C} \cdot \weight{R} \;\ge\; \weight{H}
      \label{eq:bounds:lose}
    \end{align}
The decomposition gives us a lower bound on $\weight{G}$ and an upper bound on $\weight{H}$.
Due to the common factor $\weight{C}$, for proving $\weight{G} > \weight{H}$, it suffices to show that $\weight{L} > \weight{R}$.
Thus, we only need to analyze the rule and not all rewrite steps (see Definition~\ref{def:decreasing:rule}).

Although this gives the general idea, the bounds of Equation~\eqref{eq:bounds:lose}  actually overestimate the weight. The weight of (an image of) the interface $\weight{K}$ appears as a part of $\weight{C}$, $\weight{L}$, and $\weight{R}$. As a consequence, $\weight{K}$ is counted twice in $\weight{C} \cdot \weight{L}$ and twice in $\weight{C} \cdot \weight{R}$. To avoid this overcounting, we improve the bounds as follows:
    \begin{align}
      \weight{G} \;\ge\; \weight{C - u} \cdot \weight{L} 
      &&
      \weight{C - u} \cdot \weight{R} \;\ge\; \weight{H}
      \label{eq:bounds:better}
    \end{align}
Here, $\weight{C - u}$ is the weight of $C$ excluding those morphisms $e : \domain{e} \to C$, for $e \in \elements$, that factor through $u : K \to C$ (see further Definition~\ref{def:weighing}).

\subsection{Assumptions about pushout squares}
\label{sec:assumptions:about:pushout:squares}
Finally, we need to identify assumptions about pushout squares that allow us to decompose the weights of $G$ and $H$ as in Equation~\eqref{eq:bounds:better}.
For simplicity, assume that $T$ has a single weighted element $e : X \to T$.

First, consider the right pushout square of the rewrite step together with a morphism $\alpha: X \to H$. 
The morphism $\alpha$ could potentially give rise to a commuting triangle
\begin{equation}
    \hspace{-1mm}
    \begin{tikzpicture}[baseline=-1.5ex,node distance=10mm,nodes={rectangle,inner sep=1mm,outer sep=0}]
        \node at (0,-1mm) {$=$};
        \node (G) at ($(30:7mm) + (3mm,0)$) {$H$};
        \node (T) at (-90:8mm) {$T$};
        \node (Ti) at ($(150:7mm) + (-3mm,0)$) {$X$};
        \draw [->] (G) to node [label,pos=0.4] {$\phi$} (T);
        \draw [->] (Ti) to node [label,pos=0.4] {$\alpha$} (G); 
        \draw [->] (Ti) to node [label,pos=0.4] {$e$} (T); 
    \end{tikzpicture}\;,
\end{equation}
and thereby contribute to the weight $\weight{H}$ of $H$.
For $\weight{C - u} \cdot \weight{R}$ to be an upper bound for $\weight{H}$, the morphism $\alpha$ needs to be \emph{traceable} back to $C$ or $R$. This means that, in the following diagram, at least one of the dotted morphisms must exist and form a commuting triangle (see further Definition~\ref{def:traceability}):
\begin{equation}
        \begin{tikzpicture}[node distance=14mm,leftsq/.style={black!50},rightsq/.style={opacity=1},baseline=-7mm]
          \basestep
          \node (X) [below right of=H] {$X$};
          \draw [->] (X) to node [label,pos=0.4] {$\alpha$} (H);
          \draw [->,dotted,red,line width=.3mm] (X) to[bend left=30] (C);
          \node at ($(C)!.5!(X)$) {$=$};
          \draw [->,dotted,red,line width=.3mm] (X) to[bend right=30] (R);
          \node at ($(R)!.5!(X)$) {$=$};
        \end{tikzpicture}
\end{equation}
For general pushouts in $\Graph$, for instance, any node or edge in the pushout object $H$ traces back to a node or edge in $R$ or $C$ -- the same is not true for loops in $H$, which can be created by the pushout. For pushouts along monomorphisms in $\Graph$, however, these ``traceable elements'' do include loops. 
This highlights that traceability depends on
\begin{enumerate}[label=(\alph*)]
    \item what elements we are tracing (namely, $\domain{e}$ for $e \in \elements$), as well as
    \item restrictions on the pushout squares in rewrite steps (e.g., monic $r$ in the rule).
\end{enumerate}
We will show that traceability suffices to obtain the upper bound on $\weight{H}$ (Definition~\ref{def:weighable} and Lemma~\ref{lem:decompose}).

Second, consider the left pushout square of the rewrite step together with a morphism $\alpha: X \to G$. 
For $\weight{C - u} \cdot \weight{L}$ to be a lower bound for $\weight{G}$, we must avoid overcounting in $\weight{C - u}$ and $\weight{L}$.
The problem of overcounting occurs whenever $\alpha$ traces back to more than one morphism in $L$ or $C-u$.
There are three possible scenarios, as shown in Figure~\ref{fig:overcounting:left}.
\begin{figure}[h!]
    \centering
     \subcaptionbox{Overcounting in $L$.\label{fig:overcounting:l}}[0.25\textwidth]{
     \begin{tikzpicture}[node distance=14mm,leftsq/.style={},rightsq/.style={},baseline=-7mm]
      \basestepleft
      
      \node (X) [below left of=G] {$X$};
      \draw [->] (X) to node [label,pos=0.4] {$\alpha$} (G);
      \draw [->,dotted,red,line width=.3mm] (X) to[bend left=30] (L.-140);
      \draw [->,dotted,red,line width=.3mm] (X) to[bend left=40] (L.170);
      \node at ($(L)!.5!(X)$) {$=$};
    \end{tikzpicture}
    }
    \subcaptionbox{Overcounting in $C$.\label{fig:overcounting:c}}[0.25\textwidth]{
    \begin{tikzpicture}[node distance=14mm,leftsq/.style={},rightsq/.style={},baseline=-7mm]
      \basestepleft
      
      \node (X) [below left of=G] {$X$};
      \draw [->] (X) to node [label,pos=0.4] {$\alpha$} (G);
      \draw [->,dotted,red,line width=.3mm] (X) to[bend right=30] (C.-130);
      \draw [->,dotted,red,line width=.3mm] (X) to[bend right=40] (C.-90);
      \node at ($(C)!.5!(X)$) {$=$};
    \end{tikzpicture}
    }
    \subcaptionbox{Overcounting between $L$ and $C$.\label{fig:overcounting:lc}}[0.4\textwidth]{
    \begin{tikzpicture}[node distance=14mm,leftsq/.style={},rightsq/.style={},baseline=-7mm]
      \basestepleft
      
      \node (X) [below left of=G] {$X$};
      \draw [->] (X) to node [label,pos=0.4] {$\alpha$} (G);
      \draw [->,dotted,red,line width=.3mm] (X) to[bend right=30] node [label] {$\alpha_2$} (C);
      \node at ($(C)!.5!(X)$) {$=$};
      \draw [->,dotted,red,line width=.3mm] (X) to[bend left=30] node [label] {$\alpha_1$} (L);
      \node at ($(L)!.5!(X)$) {$=$};
    \end{tikzpicture}
    }
    \caption{Causes of $\weight{C - u} \cdot \weight{L}$ overestimating the weight $\weight{G}$ of $G$.}
    \label{fig:overcounting:left}
\end{figure}

For the dotted morphisms with codomain $C$, we are interested only in those that do not factor through $u$. 
The morphisms that factor through $u$ do not contribute to overcounting since we exclude them in $\weight{C - u}$.
The three causes of overcounting can be avoided by using the following assumptions:
\begin{enumerate}[label=(\alph*)]
    \item 
        Overcounting in $L$ can be prevented by requiring the match morphism $m$ to be $X$-monic, that is, monic for morphisms with domain $X$ (see Definition~\ref{def:relative:monicity}).
        For instance, think of $X$ as an edge and of $m$ as preserving distinctness of edges. 
    \item 
        Overcounting in $C$ can be prevented by requiring the morphism $l'$ to be $X$-monic outside of $u$, that is, monic for those morphisms with domain $X$ that do not factor through $u$ (see Definition~\ref{def:relative:monicity}). 
    \item 
        Overcounting between $L$ and $C$ can be excluded by requiring vertical strong traceability (for the left pushout squares of rewrite steps), that is, the morphism $\alpha_2: X \to C$ must trace back to the interface $K$ as shown in Figure~\ref{fig:strong:traceability} (see Definition~\ref{def:traceability}).
        Then $\alpha_2$ factors through $u$ and is thus not counted in $\weight{C - u}$.
        \begin{figure}[h!]
            \centering
            \begin{tikzpicture}[node distance=14mm,leftsq/.style={},rightsq/.style={},baseline=-7mm]
              \basestepleft
              
              \node (X) [below left of=G] {$X$};
              \useasboundingbox (X.south west) rectangle (I.north east);

              \draw [->] (X) to node [label,pos=0.4] {$\alpha$} (G);
              \draw [->,dotted,red,line width=.3mm] (X) to[bend right=30] node [label] {$\alpha_2$} (C);
              \node at ($(C)!.5!(X)$) {$=$};
              \draw [->,dotted,red,line width=.3mm] (X) to[bend left=30] node [label] {$\alpha_1$} (L);
              \node at ($(L)!.5!(X)$) {$=$};

              \draw [->,dotted,red,line width=.3mm] (X) to[bend right=70,looseness=1.7] node [label] {$\alpha_2'$} (I);
              \node at ($(C)+(-2mm,-8mm)$) {$=$};

            \end{tikzpicture}
            \caption{Vertical strong traceability.}
            \label{fig:strong:traceability}
        \end{figure}
\end{enumerate}
Given the above assumptions, the weight of $G$ can be decomposed as $\weight{G} = \weight{C - u} \cdot \weight{L}$  (see further Definition~\ref{def:weighable} and Lemma~\ref{lem:decompose}). 
    
\paragraph{Remainder of the proof}
In Section~\ref{sec:weighing}, we investigate the properties of (strong) traceability in Definition~\ref{def:traceability} and relative monicity in Definition~\ref{def:relative:monicity}.
We show how to weigh morphisms rooted in pushout objects (Lemma~\ref{lem:decompose}), and then extend this to weighing pushout objects (Lemma~\ref{lem:one:side}).
In Section~\ref{sec:termination}, we analyze termination of DPO graph transformation systems. 
We propose the concept of context-closures (Definition~\ref{def:context:closure}) to ensure that every rewritable object (not in normal form) admits some morphism into the type graph $T$. 
We then introduce decreasing rules (Definition~\ref{def:decreasing:rule}) and show that these give rise to decreasing steps (Theorem~\ref{thm:steps}). This paves the way for our termination technique (Theorem~\ref{thm:termination}).

\begin{remark}[Adhesivity]
    We note that the termination technique discussed in this paper does not require general categorical properties like $\mathcal{M}$-adhesivity~\cite[Definition 2.4]{ehrig2010categorical}.
    The properties that are crucial for our technique, such as (strong) traceability and $X$-monicity, do not depend on the category only. They depend on a combination of
    \begin{enumerate}[label=(\alph*)]
        \item the category,
        \item the elements being traced/weighed,
        \item the rewrite rules, and
        \item the double pushout framework (restrictions on the pushout squares in rewrite steps).
    \end{enumerate}
    However, $\mathcal{M}$-adhesion can be used to establish strong traceability; see the paragraph before Proposition~\ref{prop:traceable:implies:strongly:traceable}.
\end{remark}

\section{Preliminaries}
\label{sec:preliminaries}

\newcommand{\dpospan}{%
        \begin{tikzpicture}[node distance=10mm,baseline=(I.base),outer sep=0,inner sep=.5mm]
          \node (I) {$K$};
          \node (L) [left of=I] {$L$};
          \node (R) [right of=I] {$R$};
          
          \draw [->] (I) to node [label] {$l$} (L);
          \draw [->] (I) to node [label] {$r$} (R);
        \end{tikzpicture}%
}

We assume an understanding of basic category theory, in particular of pushouts, pullbacks, and monomorphisms $\mono$. Throughout this paper, we work in a fixed category $\CC$.
For an introduction to DPO graph transformation, see~\cite{konig2018tutorial,ehrig1973graph,corradini1997}.

We write \Graph{} for the category of \textit{finite} multigraphs, and \SimpleGraph{} for the category of \textit{finite} simple graphs (i.e., without parallel edges). The graphs can be edge-labelled and/or node-labelled over a fixed label set.
\vspace{2ex}

\begin{definition}[DPO Rewriting~\cite{ehrig1973graph}]
    \label{def:double:pushout:rewriting}
    A \emph{DPO rewrite rule} $\rho$ is a span \dpospan. A diagram of the form
    \begin{center}
        \begin{tikzpicture}[node distance=14mm]
          \node (I) {$K$};
          \node (L) [left of=I] {$L$};
          \node (R) [right of=I] {$R$};
          \node (G) [below of=L] {$G$};
          \node (C) [below of=I] {$C$};
          \node (H) [below of=R] {$H$};
          
          \draw [->] (I) to node [label] {$l$} (L);
          \draw [->] (I) to node [label] {$r$} (R);
          \draw [->] (L) to node [label] {$m$} (G);
          \draw [->] (I) to (C);
          \draw [->] (R) to (H);
          \draw [->] (C) to (G);
          \draw [->] (C) to (H);
    
          \node [at=($(I)!.5!(G)$)] {\normalfont PO};
          \node [at=($(I)!.5!(H)$)] {\normalfont PO};
        \end{tikzpicture}
    \end{center}
    defines a \emph{DPO rewrite step} $\DPOstep{G}{H}{\rho}{m}$, i.e., a step from $G$ to $H$ using rule $\rho$ and match morphism $m : L \to G$.
\end{definition}

\begin{remark}
    \label{remark:dpo:variants}
    Usually Definition~\ref{def:double:pushout:rewriting} is accompanied with constraints. Most commonly, $l$ and $m$ are required to be (regular) monic. Sometimes also the left pushout square is restricted, e.g., it may be required that the pushout complement is \emph{minimal} or \emph{initial}~\cite{braatz2011how,behr2022fundamentals,behr2021concurrency}. Because such restrictions affect whether a DPO rewrite system is terminating, we will define our approach in a way that is parametric in the variation of DPO under consideration (see Definition~\ref{def:double:pushout:framework}).
\end{remark}

In this paper, we investigate techniques for proving \emph{strong} normalization (termination), meaning that every rewrite sequence eventually ends in a normal form (a term that can no longer be rewritten). This is in contrast to the concept of \emph{weak} normalization which only requires the existence of some rewrite sequence to a normal form.

\begin{definition}[Termination]
    Let $R,S$ be binary relations on $\ob{\categoryc}$.
    Relation $R$ is \emph{terminating (or strongly normalizing) relative to $S$}, denoted $\SN(R/S)$, if every (finite or infinite) sequence 
    \begin{align*}
        o_1 \mathrel{(R \cup S)} 
        o_2 \mathrel{(R \cup S)} 
        o_3 \mathrel{(R \cup S)} 
        o_4 \ldots &&
        \text{where $o_1,o_2,\ldots \in \ob{\categoryc}$}
    \end{align*}
    contains only a finite number of $R$ steps. 
    The relation $R$ is \emph{terminating (or strongly normalizing)}, denoted $\SN(R)$, if $R$ terminates relative to the empty relation $\varnothing$.
    
    These concepts carry over to sets of rules via the induced rewrite relations.
\end{definition}

We will interpret weights into a well-founded semiring as defined below.

\begin{definition}
    A \emph{monoid} $\tuple{S,\cdot,e}$ is a set $S$
    with an operation $\cdot : S \times S \to S$ 
    and an \emph{identity element} $e$ such that for all $a,b,c \in S$:
   \begin{align*}
       (a \cdot b) \cdot c &= a \cdot (b \cdot c) 
       &\text{and}&&  %
       a \cdot e &= a = e \cdot a
   \end{align*}
    The monoid is \emph{commutative} if 
    for all $a,b \in S$: 
    $a \cdot b = b \cdot a$.
\end{definition}

\begin{definition}
    A \emph{semiring} is a tuple $\mathcal{S} = \ring$ where 
    $S$ is a set and all of the following conditions hold:
    \begin{enumerate}[label=(\roman*)]
        \item $\tuple{S,\oplus,\ozero}$ is a commutative monoid;
        \item $\tuple{S,\otimes,\oone}$ is a monoid;
        \item $\ozero$ is an annihilator for $\otimes$: for all $a \in S$ we have
          \begin{align*}
             \ozero \otimes a &= \ozero = a \otimes \ozero \tag{annihilation}
         \end{align*}
       \item $\otimes$ distributes over $\oplus$: for all $a,b,x \in S$ we have
         \begin{align*}
             (a \oplus b) \otimes x &= (a \otimes x) \oplus (b \otimes x) \tag{right-distributivity}\\
             x \otimes (a \oplus b) &= (x \otimes a) \oplus (x \otimes b)
             \tag{left-distributivity}
         \end{align*}
    \end{enumerate}
    The semiring $\mathcal{S}$ is \emph{commutative} if $\tuple{S,\otimes,\oone}$ is commutative.
    We will often denote a semiring by its carrier set, writing $S$ for $\mathcal{S}$.
\end{definition}

\begin{definition}
    \label{def:wf:semiring}
    A \emph{\wf{} semiring} $\wfring$ 
    consists of 
    \begin{enumerate}[label=(\roman*)]
        \item 
            a semiring $\tuple{S,\oplus,\otimes,\ozero,\oone}$; and
        \item
            non-empty orders
            ${\ostrict},\, {\oweak} \subseteq S \times S$
            for which
            ${\ostrict} \subseteq {\oweak}$ and $\oweak$ is reflexive;
    \end{enumerate}
    such that $\SN({\ostricti} / {\oweaki})$, $\ozero \ne \oone$,
    and for all $x,x',y,y' \in S$ we have
    \begin{align}
      x \oweak x' \ \wedge \ y \oweak  y' &\quad\implies\quad x \oplus y \ \oweak \  x' \oplus y' \tag{S1}\label{wfring:oplus:oweak}\\
      x \ostrict x' \ \wedge \ y \ostrict y' &\quad\implies\quad x \oplus y \ \ostrict \ x' \oplus y' \tag{S2}\label{wfring:oplus:ostrict} \\
      x \oweak x' \ \wedge \ \oone \oweak y &\quad\implies\quad 
      x \otimes y \ \oweak \ x' \otimes y 
      \text{ \;and\; }
      y \otimes x \ \oweak \ y \otimes x' 
      \tag{S3}\label{wfring:otimes:oweak} \\
      x \ostrict x' \ \wedge \ \oone \oweak y \ne \ozero &\quad\implies\quad 
      x \otimes y \ \ostrict \ x' \otimes y 
      \text{ \;and\; }
      y \otimes x \ \ostrict \ y \otimes x' 
      \tag{S4}\label{wfring:otimes:ostrict} 
    \end{align}
    The semiring is \emph{strictly monotonic} if additionally 
    \begin{align}
      x \ostrict x' \ \wedge \ y \oweak y' &\quad\implies\quad x \oplus y \ \ostrict \ x' \oplus y' \tag{S5}\label{wfring:strict} 
    \end{align}
\end{definition}

\begin{proposition}
    \label{prop:wf:semiring}
    Let $\wfring$ be a \wf{} semiring, then
    \begin{align}
        \oone \oweak x,y \ \wedge \  x,y \ne \ozero &\quad\implies\quad \oone \ \oweak \ (x \otimes y) \ \ne \ \ozero \tag{S6}\label{wfring:otimes:one:zero} \\
        \oone \oweak a,b \ \wedge \ x \oweak y &\quad\implies\quad (a \oplus b) \otimes x \ \oweak \ (a \oplus b) \otimes y  \tag{S7}\label{wfring:oplus:otimes:oweak} \\
        \oone \oweak a,b \ \wedge \  a,b \ne \ozero \ \wedge \ x \ostrict y &\quad\implies\quad (a \oplus b) \otimes x \ \ostrict \ (a \oplus b) \otimes y  \tag{S8}\label{wfring:oplus:otimes:ostrict}
    \end{align}
    for all $x,y,a,b \in S$. 
\end{proposition}

\begin{proof}%
    For~\eqref{wfring:otimes:one:zero}, assume that $\oone \oweak x,y$ and $x,y \ne \ozero$.
    We get $\oone \oweak y = \oone \otimes y \oweak x \otimes y$ from~\eqref{wfring:otimes:oweak}.
    By definition, there exist $c,d \in S$ with $c \ostrict d$,
    and hence $x \otimes y = \ozero$ would contradict
    $c \otimes x \otimes y \ostrict d \otimes x \otimes y$ by~\eqref{wfring:otimes:ostrict}.

    For~\eqref{wfring:oplus:otimes:oweak}, let $\oone \oweak a,b$ and $x \oweak y$. 
    We have 
    \begin{align*}
      (a \oplus b) \otimes x 
      = (a \otimes x) \oplus (b \otimes x) 
      \oweak (a \otimes y) \oplus (b \otimes y)
      = (a \oplus b) \otimes y
    \end{align*}
    using distributivity and~\eqref{wfring:oplus:ostrict} since $a \otimes x \oweak a \otimes y$ and $b \otimes x \oweak b \otimes y$ by \eqref{wfring:otimes:oweak}.
    
    For~\eqref{wfring:oplus:otimes:ostrict}, let $\oone \oweak a,b$, $a,b \ne \ozero$ and $x \ostrict y$. We have 
    \begin{align*}
      (a \oplus b) \otimes x 
      = (a \otimes x) \oplus (b \otimes x) 
      \ostrict (a \otimes y) \oplus (b \otimes y)
      = (a \oplus b) \otimes y
    \end{align*}
    using distributivity and~\eqref{wfring:oplus:ostrict} since $a \otimes x \ostrict a \otimes y$ and $b \otimes x \ostrict b \otimes y$ by \eqref{wfring:otimes:ostrict}.
\end{proof}

\begin{example}[Arithmetic semiring]
\label{ex:natural}
    The \emph{arithmetic semiring} $\tuple{\nat,\oplus,\otimes,\ozero,\oone,\ostrict,\oweak}$, over the natural numbers $\nat$, is a strictly monotonic, well-founded semiring  where $\oplus = +$, $\otimes = \cdot$, $\ozero = 0$, $\oone = 1$, and ${\ostrict} = {<}$ and $\oweak$ are the usual orders on $\nat$.
\end{example}

\begin{example}[Tropical semiring]
\label{ex:tropical}
    The \emph{tropical semiring} $\tuple{S,\oplus,\otimes,\ozero,\oone}$ has carrier set 
    $S = \nat \cup \{\infty\}$.
    The usual order $\oweak$ on $\nat$ is extended to $S$ by letting $a \oweak \infty$ for all $a \in S$.
    The operation $+$ on $\nat$ is extended to $S$ by $a + \infty = \infty$ for all $a \in S$.
    Then the operations and unit elements of the tropical semiring are given by:
    \begin{align*}
      \oplus(a,b) &= \min \set{a,b} &
      \otimes(a,b) &= a + b & \ozero &= \infty & \oone &= 0
    \end{align*}
    The tropical semiring is a well-founded semiring 
    $\tuple{S,\oplus,\otimes,\ozero,\oone,\ostrict,\oweak}$,
    where $\ostrict$ is the strict part of~$\oweak$.
    It is not strictly monotonic because
    $2 \oplus 2 \not\ostrict 3 \oplus 2$.
\end{example}

\begin{example}[Arctic semiring]
\label{ex:arctic}
    The \emph{arctic semiring} $\tuple{S,\oplus,\otimes,\ozero,\oone}$ has carrier set 
    $S = \nat \cup \{-\infty\}$.
    The usual order $\oweak$ on $\nat$ is extended to $S$ by $-\infty \oweak a$ for all $a \in S$.
    The operation $+$ on $\nat$ is extended to $S$ by $a + (-\infty) = -\infty$ for all $a \in S$.
    The operations and unit elements of the arctic semiring are:
    \begin{align*}
      \oplus(a,b) &= \max \set{a,b} &
      \otimes(a,b) &= a + b & \ozero &= -\infty & \oone &= 0
    \end{align*}
    The arctic semiring is a well-founded semiring 
    $\tuple{S,\oplus,\otimes,\ozero,\oone,\ostrict,\oweak}$
    where $\ostrict$ is the strict part of $\oweak$.
    It is not strictly monotonic because
    $2 \oplus 3 \not\ostrict 3 \oplus 3$.
\end{example}

\section{Weighted Type Graphs}
\label{sec:typegraphs}

In this section, we introduce a notion of weighted type graphs for arbitrary categories, and we employ them to weigh morphisms (into the type graph) and objects in the category.

\begin{definition}[Weighted type graphs]\label{def:wtg}
    A \emph{weighted type graph} 
    $\awtg = \wtg$
    consists~of:
    \begin{enumerate}[label=(\roman*)]
        \item an object~$T \in \ob{\categoryc}$, called a \emph{type graph},
        \item a set $\elements$ of \emph{$T$-valued elements} $e : \domain{e} \to T$,
        \item a commutative semiring $\ring$, and
        \item a \emph{weight function} $\ww : \elements \to S$ such that, for all $e \in \elements$,\ $ \ozero \ne \ww(e) \ge \oone$.
    \end{enumerate}
    The weighted type graph~$\mathcal{T}$ is \emph{finitary} if, for every $e \in \elements$ and every $G \in \ob{\categoryc}$, 
    $\homset{\domain{e}}{G}$ and 
    $\homset{G}{T}$ are finite sets.
\end{definition}

\begin{remark}\label{rem:bruggink:gt:zero}
    Bruggink et al.~\cite{bruggink2015} introduce weighted type graphs for the category $\Graph$ of edge-labelled multigraphs.
    We generalize this concept to  arbitrary categories.    Important technical issues aside (discussed below), their notion is obtained by instantiating our notion of weighted type graphs for $\Graph$ with $\smash{\elements = \bigcup\; \set{ \homset{\edgeGraph[$x$]}{T} \mid x \text{ an edge label} }}$. 

    The technical differences are as follows. We require $\ww(e) \ge \oone$ and $\ww(e) \ne \ozero$ 
    for every $e \in \elements$.
    The work~\cite{bruggink2015} requires 
    \begin{itemize}
        \item $\ww(e) > \ozero$ only for loops $e \in \elements$ on a distinguished `flower node'. 
    \end{itemize}
    This condition is insufficient, causing \cite[Theorem 2]{bruggink2015} to fail as illustrated by \cite[Example A.1]{gwtg2024arxiv}.  
    In the revised version~\cite{bruggink2023arxiv}, this problem has been fixed.

    Moreover, the condition $\ww(e) > \ozero$ rules out the use of the tropical semiring as $\ozero = +\infty$. 
    Observe that $\ww(e) \ge \oone$ is often a weaker condition than $\ww(e) > \ozero$. In particular, our condition enables the use of the tropical semiring (for instance, see Example~\ref{ex:monic:simple}). 
\end{remark}

\begin{notation}
    We use the following notation:
    \begin{enumerate}[label=(\alph*)]
    \item
        For $\alpha : A \to B$, $x : A \to C$, let 
        \(
           \set{ {-} \circ \alpha = x } \;=\; \set{ \phi : B \to C \mid \phi \circ \alpha = x } 
        \);
    \item        
        For $\alpha : B \to C$, $x : A \to C$, let 
        \(
           \set{ \alpha \circ {-} = x } \;=\; \set{ \phi : A \to B \mid \alpha \circ \phi = x } 
        \);
    \item
        For $\alpha : B \to C$ and a
        set $S$ %
        of morphisms $\phi: A \to B$, we define
        \(
            \alpha \circ S \;=\; \set{ \alpha \circ \phi \mid \phi \in S }
        \).
    \end{enumerate}
\end{notation}

\begin{definition}[Weighing morphisms and objects]\label{def:weighing}
    Let $\awtg = \wtg$ be a finitary weighted type graph and let $G,A \in \ob{\categoryc}$.
    \begin{enumerate}[label=(\roman*),itemsep=1ex]
        \item
            The \emph{weight of a morphism $\phi : G \to T$} is\,
            $\displaystyle\weight[\mathcal{T}]{\phi} = \oprod_{e \in \elements} \ww(e)^{\counting{e}{\phi}}$.\\
            This is equivalent to
            \begin{math}\displaystyle
                \weight[\mathcal{T}]{\phi} %
                \;= \hspace{-5mm}\oprod_{\bigsub{e \in \elements\\ \alpha \in \set{ {\phi \circ {-}} = e }}} \hspace{-5mm}\ww(e)
            \end{math};
            visually
            \begin{tikzpicture}[baseline=(Ti.base),node distance=11mm,nodes={rectangle,inner sep=1mm,outer sep=0}]
                \node at (0,0) {$=$};
                \node (G) at ($(30:7mm) + (3mm,0)$) {$G$};
                \node (T) at (-90:5mm) {$T$};
                \node (Ti) at ($(150:7mm) + (-3mm,0)$) {$\domain{e}$};
                \draw [->] (G) to node [label,pos=0.4] {$\phi$} (T);
                \draw [->] (Ti) to node [label,pos=0.4] {$\alpha$} (G); 
                \draw [->] (Ti) to node [label,pos=0.4] {$e$} (T); 
            \end{tikzpicture}.

        \item
            The \emph{weight of finite $\Psi \subseteq \set{ \phi \mid \phi : G \to T}$} is 
                $\weight[\mathcal{T}]{\Psi} = \osum_{\phi \in \Psi} \weight[\mathcal{T}]{\phi}$.
        \item
            The \emph{weight of an object $G \in \ob{\categoryc}$} is 
                $\weight[\mathcal{T}]{G} = 
                \weight[\mathcal{T}]{\homset{G}{T}}$.
        \item
            The \emph{weight of $\phi : G \to T$ excluding $(\alpha \circ {-})$}, for $\alpha : A \to G$, is 
            \begin{align*}
                \weight[\mathcal{T}]{\phi - (\alpha \circ {-})} 
                &\;=\; \oprod_{e \in \elements} \ww(e)^{\countingmin{e}{\phi}{\alpha}{\zeta}}
            \end{align*}
            In other words, we count all $\tau$ such that for all $\zeta : \domain{e} \to A$, we have the following:
            \begin{center}
                \begin{tikzpicture}[baseline=1ex,node distance=13mm,nodes={rectangle,inner sep=1mm,outer sep=0}]
                    \node (G) {$G$};
                    \node (T) [left of=G,yshift=-9mm] {$T$};
                    \node (domX) [left of=G,node distance=26mm] {$\domain{e}$};
                    \node (A) [left of=G,yshift=9mm] {$A$};
                    \node [left of=G,yshift=4mm] {$\neq$};
                    \node [left of=G,yshift=-4mm] {$=$};
                    \draw [->] (G) to node [label,pos=0.4] {$\phi$} (T);
                    \draw [->] (A) to node [label,pos=0.4] {$\alpha$} (G); 
                    \draw [->] (domX) to node [label,pos=0.4] {$e$} (T); 
                    \draw [->] (domX) to node [label,pos=0.5] {$\tau$} (G); 
                    \draw [->] (domX) to node [label,pos=0.5]{$\zeta$} (A); 
                \end{tikzpicture}
            \end{center}

    \end{enumerate}
\end{definition}

Observe the analogy to weighted automata where the weight of a word is the sum of the weights of all runs, and the weight of a run is the product of the edge weights.

\begin{lemma}\label{lem:weight:one}
    Let $\awtg = \wtg$ be a finitary weighted type graph. Then
    \begin{enumerate}
        \item for every $\phi : G \to T$: $\oone \oweak \weight[\mathcal{T}]{\phi} \ne \ozero$;
        \item for every $\phi : G \to T$ and $\alpha : A \to G$: $\oone \oweak \weight[\mathcal{T}]{\phi - (\alpha \circ {-})} \ne \ozero$;
    \end{enumerate}
\end{lemma}

\begin{proof}%
    Follows from $\oone \oweak \ww(e) \ne \ozero$ for every $e \in \elements$ and Proposition~\ref{prop:wf:semiring}.
\end{proof}

\section{Weighing Pushouts}
\label{sec:weighing}

In this section, we introduce techniques for determining exact and upper bounds on the weights of pushout objects in terms of the corresponding pushout span. For these techniques to work, we need certain assumptions on the pushout squares.
This concerns, in particular, traceability of elements (Section~\ref{sec:traceability}) and relative monicity conditions (Section~\ref{sec:monicity}), which are used to define the notion of weighable pushout squares (Section~\ref{sec:weighable}). Using the notion of a weighable pushout square, we can bound weights of morphisms out of the pushout object (Section~\ref{sec:weighing:morphisms}), and consequently bound the weight of the pushout object itself (Section~\ref{sec:sec:weighing:objects}).

The following definition introduces an orientation for pushout squares, distinguishing between horizontal and vertical morphisms. A double pushout rewrite step with respect to a rule consists of two oriented pushout squares. For both squares, $\beta$ is the shared morphism, and $\alpha$ will be $l$ or $r$ of the rule span, respectively. For the left square, $\beta'$ is the match morphism. 
The distinction between horizontal and vertical morphisms will be used in Definitions~\ref{def:traceability} and~\ref{def:weighable} below to impose different conditions on these morphisms.

\begin{definition}[Oriented pushout squares]\label{def:oriented:pushout}
    An \emph{oriented pushout square} $\tau$ is a pushout square with a fixed orientation:
      \begin{center}
        \begin{tikzpicture}[node distance=14mm]
          \node (I) {$A$};
          \node (R) [right of=I] {$B$};
          \node (C) [below of=I] {$C$};
          \node (H) [below of=R] {$D$};
          
          \draw [->] (I) to node [label] {$\alpha$} (R);
          \draw [->] (I) to node [label] {$\beta$} (C);
          \draw [->] (R) to node [label] {$\beta'$} (H);
          \draw [->] (C) to node [label] {$\alpha'$} (H);
    
          \node [at=($(I)!.5!(H)$)] {\normalfont PO};
        \end{tikzpicture}
      \end{center}
    We distinguish the morphisms into horizontal and vertical. Morphisms $\beta$ and $\beta'$ are \emph{vertical} and $\alpha$ and $\alpha'$ are \emph{horizontal}.
    
    We denote the pushout square by 
    $\posquare{A}{\alpha}{B}{\beta}{C}{\alpha'}{\beta'}{D}$ where the order of the morphisms in the span is vertical before horizontal, and the order of the morphisms in the cospan is horizontal before vertical.
\end{definition}

\subsection{Traceability}
\label{sec:traceability}

We introduce the concept of traceability of elements $X \in \ob{\categoryc}$ along pushout squares. Roughly speaking, traceability guarantees that the pushout cannot create occurrences of $X$ out of thin air. 

\begin{definition}[Traceability]
    \label{def:traceability}
    Let $\Delta$ be a class of oriented pushout squares.
    An object $X \in \ob{\categoryc}$ is called \emph{traceable along $\Delta$} if: whenever we have a configuration of the form
    \begin{center} 
        \begin{tikzpicture}[node distance=14mm,baseline=(D)]
          \node (A) {$A$};
          \node [right of=A] (B) {$B$}; 
          \node [below of=A] (C) {$C$}; 
          \node [below of=B] (D) {$D$}; 
          \node [below right of=D] (X) {$X$};
          
          \begin{scope}[nodes=rectangle]          
          \draw [->] (A) to node [label,pos=0.5] {$\alpha$} (B);
          \draw [->] (A) to node [label,pos=0.5] {$\beta$} (C);
          \draw [->] (B) to node [label,pos=0.45] {$\beta'$} (D); 
          \draw [->] (C) to node [label,pos=0.45] {$\alpha'$} (D);
          \draw [->] (X) to node [label,pos=0.4] {$f$} (D);
          \end{scope}
          
          \node at ($(A)!.5!(D)$) {$\Delta$};
        \end{tikzpicture}\;, 
    \end{center}
    such that the displayed square is from $\Delta$, then 
    \begin{enumerate}[label=(\alph*)]
        \item\label{traceable:a} there exists a $g : X \to B$ such that $f = \beta'g$, or
        \item\label{traceable:b} there exists an $h : X \to C$ such that $f = \alpha'h$.
    \end{enumerate}
    Moreover we say that 
    \begin{enumerate}[label=(\alph*),resume]
        \item $X$ is \emph{vertically strongly traceable along $\Delta$} if:\\ whenever \ref{traceable:a} and \ref{traceable:b} hold together, then there exists $h' : X \to A$ such that $h = \beta h'$; \footnote{For our  termination technique we only need vertical strong traceability. However, in all our examples, strong traceability holds horizontally and vertically.}\footnote{
        In~\cite{gwtg2024icgt} mistakenly a weaker condition has been stated (requiring only the existence of $f' : X \to A$ with $f = \beta' \alpha f'$). This weaker condition is insufficient as it causes a problem with the property~($\star$) in the proof of Lemma~\ref{lem:decompose}.}

        \item $X$ is \emph{horizontally strongly traceable along $\Delta$} if:\\ whenever \ref{traceable:a} and \ref{traceable:b} hold together, then there exists $g' : X \to A$ such that $g = \alpha g'$;
        \item  $X$ is \emph{strongly traceable along $\Delta$} if:\\ $X$ is both horizontally and vertically strongly traceable along $\Delta$.%
    \end{enumerate}
\end{definition}

In $\mathcal{M}$-adhesive categories (for $\mathcal{M}$ a stable class of monos)~\cite[Definition 2.4]{ehrig2010categorical}, pushouts along $\mathcal{M}$-morphisms are also pullbacks~ \cite[Lemma 13]{lack2004adhesive}. This is a common setting for graph transformation. In this case, the notions of traceability and strong traceability coincide for pushouts along $\mathcal{M}$-morphisms, as the proposition below shows.

\begin{proposition}
    \label{prop:traceable:implies:strongly:traceable}
    If $\Delta$ is a class of pushouts that are also pullbacks, then traceability along $\Delta$ implies strong traceability along $\Delta$ (for $X \in \ob{\catname{C}}$).
\end{proposition}

\begin{proof}%
    If conditions \ref{traceable:a} and \ref{traceable:b} of Definition~\ref{def:traceability} hold, the outer diagram of
    \begin{center}
        \begin{tikzcd}
            X \arrow[rr, "g" description, bend left] \arrow[rd, "h" description] \arrow[r, "i" description, dotted] & A \arrow[r, "\alpha" description] \arrow[d, "\beta" description] \arrow[rd, "\delta", phantom] & B \arrow[d, "\beta'" description] \\
            & C \arrow[r, "\alpha'" description]                          & D                                
        \end{tikzcd}
    \end{center}
    commutes, so that $i$ is obtained (uniquely) by the pullback property of $\delta \in \Delta$, making the respective triangles commute.
\end{proof}

\begin{remark}[Traceability in \Graph]\label{rem:traceable:graph}
    In \Graph, the objects $\nodeGraph$ and $\edgeGraph[$x$]$ (for edge labels $x$) are the only (non-initial) objects that are (strongly) traceable along all pushout squares. 
    Other objects, such as loops~$\loopGraph$, are strongly traceable if all morphisms in the square are monomorphisms.
\end{remark}

\begin{remark}[Traceability in \SimpleGraph]\label{rem:traceable:simplegraph}
    In the category of simple graphs \SimpleGraph, object $\nodeGraph$ is strongly traceable along all pushout squares. An object $X = \edgeGraph[$x$]$ (for an edge label $x$) is traceable along all pushout squares, but not necessarily strongly traceable. 

    The following pushout is a counterexample:
    \begin{center}
        \begin{tikzpicture}[node distance=12mm,baseline=-6mm]
          \node (A) {$\nodeGraph$};
          \node (B) [right of=A] {$\nodeGraph$};
          \node (C) [below of=A] {$\nodeGraph$};
          \node (D) [below of=B] {$\nodeGraph$};
          
          \draw [mono] (B) to (A);
          \draw [mono] (B) to (D);
          \draw [double,double distance=.5mm] (C) to (A);
          \draw [double,double distance=.5mm] (C) to (D);
    
          \draw [->,loop=180,distance=1.3em] (A) to node [label] {$x$} (A);
          \draw [->,loop=0,distance=1.3em] (D) to node [label] {$x$} (D);
          \draw [->,loop=180,distance=1.3em] (C) to node [label] {$x$} (C);
    
          \node [at=($(A)!.5!(D)$)] {\normalfont PO};
        \end{tikzpicture}
    \end{center}
    However, using Proposition~\ref{prop:traceable:implies:strongly:traceable}, we have that $X$ is strongly traceable along pushout squares in which one of the span morphisms $\alpha$ or $\beta$ is a regular monomorphism, because such pushouts are pullbacks in \SimpleGraph (and more generally in quasitoposes)~\cite[Proposition 10 \& Corollary 18]{johnstone2007quasitoposes}.
\end{remark}

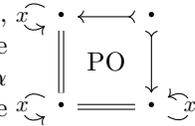
\newcommand{\graphIndex}{%
  \begin{tikzpicture}[baseline=-.6ex]
    \node (x) [circle,inner sep=0,outer sep=1mm,minimum size=0.7mm] {E};
    \node (y) [circle,inner sep=0,outer sep=1mm,minimum size=0.7mm,, right of=x, xshift=-2mm] {V};
    \draw [->] ([yshift=.7mm]x.east) to ([yshift=.7mm]y.west);
    \draw [->] ([yshift=-.7mm]x.east) to ([yshift=-.7mm]y.west);
  \end{tikzpicture}\,%
}

\begin{conjecture}
    As is well known, the category of unlabeled graphs is equivalent to the functor category $[\graphIndex, \Set]$, where index category \graphIndex is the category consisting of two objects $E$ and $V$ and two non-identity arrows that are parallel. This makes the category of unlabeled graphs a presheaf category, and the two representable functors for this category are precisely the one node graph $\cdot$ and the one edge graph $\cdot \to \cdot$ (see \cite[Section 3]{vigna2003graphistopos} for more detail). We identified these two graphs as non-initial objects that are strongly traceable along any pushout (Remark~\ref{rem:traceable:graph}). Our conjecture is that for any presheaf category for which the index category is small and acyclic, the non-initial strongly traceable objects for any pushout are precisely the representable functors. For many examples of such presheaf categories, see the overview provided by L\"owe~\cite[Section 3.1]{lowe1993algebraic} (there called graph structures).
\end{conjecture}

\subsection{Relative Monicity}
\label{sec:monicity}

Traceability suffices to derive an upper bound on the weight of a pushout. However, we need additional monicity assumptions for lower or exact bounds.

\begin{definition}[Relative monicity]\label{def:relative:monicity}
    A morphism $f : A \to B$ is called 
    \begin{enumerate}[label=(\roman*)]
        \item 
            \emph{monic for $S \subseteq \homset{-}{A}$} 
            if $fg = fh \implies g = h$ for all $g,h \in S$,
        \item 
            \emph{$X$-monic}, where $X \in \ob{\catname{C}}$, if $f$ is monic for $\homset{X}{A}$, and
        \item 
            \emph{$X$-monic outside of $u$}, where $X \in \ob{\catname{C}}$ and $u : C \to A$, if $f$ is monic for $\homset{X}{A} - u \circ \homset{X}{C}$.
    \end{enumerate}
    For $\Gamma \subseteq \ob{\categoryc}$, $f$ is called \emph{$\Gamma$-monic} if $f$ is $X$-monic for every $X \in \Gamma$.
    When~$\Gamma$ is clear from the context, %
        \begin{tikzpicture}[baseline=(A.base)]
          \node (A) {$A$};
          \node (B) [right of=A] {$B$};
          \draw [xmonic] (A) to (B);
        \end{tikzpicture}%
    denotes a $\Gamma$-monic morphism from $A$ to $B$.
\end{definition}

\begin{example}[Edge-Monicity]
    \label{example:edge:monicity}
    In \Graph, let $\Gamma = \{ \edgeGraph[$x$] \mid x \text{ is an edge label} \}$. A morphism $f: G \to H$ that does not identify edges (but may identify nodes) is not necessarily monic, but it is $\Gamma$-monic. In this particular case for \Graph, we will say that $f$ is \emph{edge-monic}.
\end{example}

\subsection{Weighable Pushout Squares}
\label{sec:weighable}
The following definition summarizes the conditions required to derive an upper or exact bound on the weight of a pushout object.

\begin{definition}[Weighable pushout square]\label{def:weighable}
    Let $\awtg = \wtg$ be a finitary weighted type graph.
    Consider an oriented pushout square $\delta$:
    \begin{center}{\normalfont
        \begin{tikzpicture}[node distance=14mm]
            \node (A) {$A$};
            \node (B) [right of=A] {$B$};
            \node (C) [below of=A] {$C$};
            \node (D) [right of=C] {$D$};
            
            \draw [->] (A) to node [label] {$\alpha$} (B);
            \draw [->] (A) to node [label] {$\beta$} (C);
            \draw [->] (B) to node [label,pos=0.45] {$\beta'$} (D);
            \draw [->] (C) to node [label,pos=0.45] {$\alpha'$} (D);
            
            \node [at=($(A)!.5!(D)$)] {$\delta$};
        \end{tikzpicture}
    }\end{center}
    We say that $\delta$ is
    \begin{enumerate}[label=(\alph*)]
        \item \emph{weighable} with $\awtg$ if:
            \begin{enumerate}[label=(\roman*)]
                \item $\domain{\elements}$ is vertically strongly traceable along $\delta$,
                \item $\beta'$ is $\domain{\elements}$-monic, and
                \item $\alpha'$ is $\domain{\elements}$-monic outside of $\beta$.
            \end{enumerate}
        \item \emph{bounded above} by $\awtg$ if $\domain{\elements}$ is traceable along $\delta$.
    \end{enumerate}
\end{definition}

\subsection{Weighing Morphisms Rooted in Pushout Objects}
\label{sec:weighing:morphisms}

Consider the pushout diagram in Lemma~\ref{lem:decompose}, below.
The lemma decomposes the weight $\weight[\mathcal{T}]{\phi}$ of the morphism $\phi$, rooted in the pushout object, into the product of the weights  $\weight[\mathcal{T}]{\phi \circ \beta'}$ and $\weight[\mathcal{T}]{\phi \circ \alpha' - (\beta \circ {-})}$.
Intuitively, the reasoning is as follows.
Let $e \in \elements$ be a weighted element. Assume that $e : \domain{e} \to T$ occurs in $\phi$, that is, there exists $d : \domain{e} \to D$ with $e = \phi \circ d$.
Then  traceability~guarantees that this occurrence $d$ can be traced back to $\phi \circ \beta'$ or $\phi \circ \alpha'$.
In case the occurrence $d$ traces back all the way to $A$ (to $\phi \circ \beta' \circ \alpha = \phi \circ \alpha' \circ \beta$), then it occurs in both $\phi \circ \beta'$ and $\phi \circ \alpha'$. To avoid unnecessary double counting, we exclude those elements from $\phi \circ \alpha'$, giving rise to $\weight[\mathcal{T}]{\phi \circ \alpha' - (\beta \circ {-})}$.
In this way, we obtain the following upper bound on the weight of $\phi$:
\begin{align*}
   \weight[\mathcal{T}]{\phi} \ \oweak \ \weight[\mathcal{T}]{\phi \circ \beta'} \otimes \weight[\mathcal{T}]{\phi \circ \alpha' - (\beta \circ {-})}
\end{align*}
To obtain the precise weight we need to avoid double counting.
First, we need vertical strong traceability to exclude all double counting between $\phi \circ \beta'$ and $\phi \circ \alpha' - (\beta \circ {-})$. Second, we need $\domain{\elements}$-monicity of $\alpha'$ outside of $\beta$ and $\domain{\elements}$-monicity of $\beta'$ to avoid double counting along these morphisms.

\begin{lemma}[Weighing morphisms rooted in pushout objects]\label{lem:decompose}
    Let a finitary weighted type graph $\awtg = \wtg$ be given.
    Consider an oriented pushout square $\delta$:
    \begin{center}{\normalfont
        \begin{tikzpicture}[node distance=14mm]
            \node (A) {$A$};
            \node (B) [right of=A] {$B$};
            \node (C) [below of=A] {$C$};
            \node (D) [right of=C] {$D$};
            \node (T) [right of=D] {$T$};
            
            \draw [->] (A) to node [label] {$\alpha$} (B);
            \draw [->] (A) to node [label] {$\beta$} (C);
            \draw [->] (B) to node [label,pos=0.45] {$\beta'$} (D);
            \draw [->] (C) to node [label,pos=0.45] {$\alpha'$} (D);
            \draw [->] (D) to node [label,pos=0.4] {$\phi$} (T);
            
            \node [at=($(A)!.5!(D)$)] {$\delta$};
        \end{tikzpicture}
    }\end{center}
    We define
       $k = \weight[\mathcal{T}]{\phi \circ \beta'} \otimes \weight[\mathcal{T}]{\phi \circ \alpha' - (\beta \circ {-})}$.
    Then the following holds: 
    \begin{enumerate}[label=(\Alph*)]
        \item\label{weighing:rooted:a}
            We have $\weight[\mathcal{T}]{\phi} = k$ if $\delta$ is weighable with $\awtg$.
        \item\label{weighing:rooted:b}
            We have $\weight[\mathcal{T}]{\phi} \le k$ if $\delta$ is bounded above by $\awtg$.
    \end{enumerate}
\end{lemma}

\begin{proof}%
    Recall the relevant definitions:
    \begin{align*}
        \weight[\mathcal{T}]{\phi} 
        &= \oprod_{e \in \elements} \ww(e)^{\counting{e}{\phi}}\\
        \weight[\mathcal{T}]{\phi \circ \beta'} 
        &= \oprod_{e \in \elements} \ww(e)^{\counting{e}{\phi \circ \beta'}} \\
        \weight[\mathcal{T}]{\phi \circ \alpha' - (\beta \circ {-})}
        &= \oprod_{e \in \elements} \ww(e)^{\countingmin{e}{\phi \circ \alpha'}{\beta}{\zeta}}
    \end{align*}
    For part \ref{weighing:rooted:a} of the lemma, it suffices to prove that
    \begin{gather}
    \begin{aligned}
        \counting{e}{\phi}
        =&\ \counting{e}{\phi \circ \beta'} + {}\\
         &\ \countingmin{e}{\phi \circ \alpha'}{\beta}{\zeta}
    \end{aligned}\label{eq:sums}
    \end{gather}
    for every $e \in \elements$.
    Moreover, for part \ref{weighing:rooted:b} of the lemma, it suffices to prove Equation~\eqref{eq:sums} with $\le$ instead of $=$ (because by definition $\oone \le \ww(e)$ for every $e \in \elements$, any additional terms on the right do not decrease the product by~\eqref{wfring:otimes:oweak}).
    
    Let $e \in \elements$ be arbitrary.
    We define morphism sets
    \begin{align*}
        \Psi_A &= \set{\phi \circ \beta' \circ \alpha \circ {-} = e} &
        \Psi_B &= \set{\phi \circ \beta' \circ {-} = e} &
        \Psi_D &= \set{\phi \circ {-} = e} \\
        &&
        \Psi_C &= \set{\phi \circ \alpha' \circ {-} = e}
    \end{align*}
    Observe that by $\alpha' \circ \beta = \beta' \circ \alpha$, we have 
    \begin{align*}
        \Psi_A = \set { \phi \circ \beta' \circ \alpha \circ {-} = e } = \set { \phi \circ \alpha' \circ \beta \circ {-} = e }
    \end{align*}
    and so $\beta \circ \Psi_A 
    = \set { \tau \in \Psi_C \mid \exists \zeta.\ \tau = \beta \circ \zeta }$. Equation~\eqref{eq:sums} thus becomes
    \begin{align}
        \# \Psi_D &= \# \Psi_B + \# \Psi_C  - \#( \beta \circ \Psi_A)
        \label{eq:sums:short:beta}
    \end{align}
    Because $\domain{\elements}$ is traceable along $\delta$ and the square $\delta$ commutes, we have
    \begin{align}
        \Psi_D 
        =&\ \beta' \circ \Psi_B \cup \alpha' \circ \Psi_C 
        \label{eq:BC:cup} \\
        &\ \beta' \circ \Psi_B \cap \alpha' \circ \Psi_C
        \supseteq \alpha' \circ \beta \circ \Psi_A
        \label{eq:BC:cap}
    \end{align}
    For part (B) of the lemma, we
    use the general laws
    \begin{align}
       \#(A \cup B) &\le \#A + \#B \label{eq:gen:cup} \\
       \#(f \circ A) &\le \#A \label{eq:gen:morph} \\
       (f \circ A - f \circ B) &\subseteq f\circ (A - B) \label{eq:gen:min} 
    \end{align}
    and proceed as follows:
    \begin{align}
        \Psi_D 
        &= \beta' \circ \Psi_B \cup (\alpha' \circ \Psi_C - \alpha' \circ \beta \circ \Psi_A) 
        &&\text{by \eqref{eq:BC:cup} and \eqref{eq:BC:cap}}
        \label{eq:D:min}\\
        \# \Psi_D 
        &\le \#( \beta' \circ \Psi_B ) + \#( \alpha' \circ \Psi_C - \alpha' \circ \beta \circ \Psi_A ) 
        &&\text{by \eqref{eq:D:min} and \eqref{eq:gen:cup}} \\
        &\le \# \Psi_B + \#( \alpha' \circ \Psi_C - \alpha' \circ \beta \circ \Psi_A ) 
        &&\text{by \eqref{eq:gen:morph}} \\
        &\le \# \Psi_B + \#( \alpha' \circ (\Psi_C - \beta \circ \Psi_A) ) 
        &&\text{by \eqref{eq:gen:min}} \\
        &\le \# \Psi_B + \#( \Psi_C - \beta \circ \Psi_A ) 
        &&\text{by \eqref{eq:gen:morph}} \\
        &= \# \Psi_B + \# \Psi_C - \#( \beta \circ \Psi_A)
        \label{eq:weighing:part:b}
        &&\text{since $\beta \circ \Psi_A \subseteq \Psi_C$}
    \end{align}
    Equation~\eqref{eq:weighing:part:b} establishes Equation~\eqref{eq:sums:short:beta} for $\le$, and hence proves part \ref{weighing:rooted:b}.

    We now prove part (A) of the lemma.
    Because $\domain{\elements}$ is vertically strongly traceable along $\delta$, we have
    $\beta' \circ \Psi_B \cap \alpha' \circ \Psi_C = \alpha' \circ \beta \circ \Psi_A$. Hence from~\eqref{eq:BC:cup} we obtain
    \begin{align}
        \Psi_D 
        &= \beta' \circ \Psi_B \uplus (\alpha' \circ \Psi_C - \alpha' \circ \beta \circ \Psi_A)
        \label{eq:sets:intermediate} \\
        \# \Psi_D 
        &= \#( \beta' \circ \Psi_B ) + \#( \alpha' \circ \Psi_C - \alpha' \circ \beta \circ \Psi_A ) \\
        &= \# \Psi_B + \#( \alpha' \circ \Psi_C - \alpha' \circ \beta \circ \Psi_A ) 
        \label{eq:sums:intermediate:size:1}
    \end{align}
    where the last equality is justified by $\domain{\elements}$-monicity of $\beta'$.
    
    We have the following property:
    \begin{itemize}\setlength{\itemindent}{1.5cm}
    \smallskip
        \item [($\star$)]
            For any $f,g \in \Psi_C$ with $f \ne g$ and $g \not\in \beta \circ \Psi_A$, we have $\alpha' \circ f \ne \alpha' \circ g$.
    \end{itemize}
    The validity of this property can be argued as follows:
    \begin{enumerate}[label=(\roman*)]
        \item If $f \not\in \beta \circ \Psi_A$, then $\alpha' \circ f \ne \alpha' \circ g$ follows since $\alpha'$ is $\domain{\elements}$-monic outside of $\beta$.
        \item If $f \in \beta \circ \Psi_A$, then there exists $f' \in \Psi_A$ with $f = \beta \circ f'$.
        For a contradiction, assume that $\alpha' \circ f = \alpha' \circ g$.
        Then $\alpha' \circ g = \alpha' \circ f = \alpha' \circ \beta \circ f' = \beta' \circ \alpha \circ f'$. From vertical strong traceability, it follows that $g \in \beta \circ \Psi_A$. This contradicts our assumption.
    \end{enumerate}    
    Thus, we have established property ($\star$).
    From ($\star$) it follows that
    \begin{align}
        \#( \alpha' \circ \Psi_C - \alpha' \circ \beta \circ \Psi_A )
        &= 
        \#( \alpha' \circ (\Psi_C - \beta \circ \Psi_A) )\\
        &= 
        \# \Psi_C - \# (\beta \circ \Psi_A)
    \end{align}
    Then using Equation~\eqref{eq:sums:intermediate:size:1}, we get
    \begin{align}
        \# \Psi_D 
        &= \# \Psi_B + \#( \alpha' \circ \Psi_C - \alpha' \circ \beta \circ \Psi_A )
        \label{eq:sums:intermediate:size:2}\\
        &= \# \Psi_B + \# \Psi_C - \# (\beta \circ \Psi_A)
    \end{align}
    This establishes Equation~\eqref{eq:sums:short:beta}, and hence proves part \ref{weighing:rooted:a}.
\end{proof}

\begin{remark}[Related work]
    Lemma~\ref{lem:decompose} is closely related to Lemma~2 of~\cite{bruggink2015} which establishes
       $\weight[\mathcal{T}]{\phi} 
       = \weight[\mathcal{T}]{\phi \circ \beta'} \otimes \weight[\mathcal{T}]{\phi \circ \alpha'}$ 
    for directed, edge-labelled multigraphs under the following assumptions:
    \begin{enumerate}[label=(\roman*)]
        \item the interface $A$ must be discrete (contains no edges),
        \item the weighted type graph has only edge weights (no node weights).
    \end{enumerate}
    As a consequence of the discrete interface $A$, all arrows involved in the pushout are edge-monic (Example~\ref{example:edge:monicity}). 

    Lemma~\ref{lem:decompose} generalizes Lemma~2 of~\cite{bruggink2015} in multiple directions:
    \begin{enumerate}[label=(\alph*)]
        \item it works for arbitrary categories with some (strongly) traceable $\domain{\elements}$, 
        \item it does not require discreteness ($\domain{\elements}$-freeness) of the interface, and
        \item it allows for $\alpha$ and $\alpha'$ that are not $\domain{\elements}$-monic.
    \end{enumerate}
    Both (a) and (b) are crucial to make the termination techique applicable in general categories, and (c) is important to allow for rewrite rules whose right morphism merges elements that we are weighing.
\end{remark}

\subsection{Weighing Pushout Objects}
\label{sec:sec:weighing:objects}

The following lemma (cf.~\cite[Lemma~1]{bruggink2015}) is a direct consequence of the universal property of pushouts. We will use this property for weighing pushout objects.

\begin{lemma}[Pushout morphisms]
\label{lem:bijection}
  Let $T \in \ob{\categoryc}$ and consider a pushout square
  \begin{center}
        \begin{tikzpicture}[node distance=14mm,nodes={rectangle,inner sep=0.5mm
        },baseline=-4.5mm]
            \node (A) [] {$A$};
            \node (B) [right of=A] {$B$};
            \node (C) [below of=A] {$C$};
            \node (D) [right of=C] {$D$};
            \node at ($(A)!.5!(D)$) {$\delta$};
            
            \draw [->] (A) to node [label] {$\alpha$} (B);
            \draw [->] (A) to node [label] {$\beta$} (C);
            \draw [->] (B) to node [label,pos=0.45] {$\beta'$} (D);
            \draw [->] (C) to node [label,pos=0.45] {$\alpha'$} (D);
        \end{tikzpicture}
  \end{center}
  For every morphism $t_A : A \to T$, there is a bijection $\theta$ from
  \begin{enumerate}[label=(\alph*)]
    \item the class $U$ of morphisms $t_D : D \to T$ for which $t_A = t_D \circ \beta' \circ \alpha$, to 
    \item the class $P$ of pairs $(t_B : B \to T,\; t_C : C \to T)$ of morphisms for which $t_A = \compi{\alpha}{t_B} = \compi{\beta}{t_C}$,
  \end{enumerate}
  such that for all $t_D \in U$ and pairs $(t_B, t_C) \in P$, we have:
  \begin{align*}
    \theta(t_D) = (t_B, t_C) \; \iff \; t_B = t_D \circ \beta' \text{ and } t_C = t_D \circ \alpha'\;.
  \end{align*}
\end{lemma}

\begin{proof}%
  Let $t_A : A \to T$ be given. For $p = (t_B, t_C) \in P$, define $\zeta(p) = t_D$ where $t_D : D \to T$ is the unique morphism satisfying $t_B = {t_D \circ \beta'}$ and $t_C = {t_D \circ \alpha'}$. Then $t_A = t_B \circ \alpha = t_D \circ \beta' \circ \alpha$ and so $\zeta(p) \in U$. Map $\zeta$ is injective, because if $\zeta(t_B, t_C) = t_D = \zeta(t_B', t_C')$, then $t_B = t_D \circ \beta' = t_B'$ and $t_C = t_D \circ \alpha' = t_C'$. It is surjective, because for any $t_D \in U$, $(t_D \circ \beta', t_D \circ \alpha') \in P$ and $t_D$ must be the universal morphism for this pair. Thus $\zeta$ and $\theta = \zeta^{-1}$ are bijections. Direction $\implies$ now trivially follows from the commutative property of universal morphisms, and direction $\impliedby$ by using uniqueness of universal morphisms.
\end{proof}

\begin{lemma}[Weighing pushout objects]\label{lem:one:side}
    Let $\awtg = \wtg$ be a finitary weighted type graph.
    Consider an oriented pushout square $\delta$:
    \begin{center}
        \begin{tikzpicture}[node distance=14mm,nodes={rectangle,inner sep=0.5mm
        },baseline=-4.5mm]
            \node (A) [] {$A$};
            \node (B) [right of=A] {$B$};
            \node (C) [below of=A] {$C$};
            \node (D) [right of=C] {$D$};
            \node at ($(A)!.5!(D)$) {$\delta$};
            
            \draw [->] (A) to node [label] {$\alpha$} (B);
            \draw [->] (A) to node [label] {$\beta$} (C);
            \draw [->] (B) to node [label,pos=0.45] {$\beta'$} (D);
            \draw [->] (C) to node [label,pos=0.45] {$\alpha'$} (D);
        \end{tikzpicture}
    \end{center}
    Define
      \begin{align*}
        k
        &= \osum_{t_A : A \to T} \; 
           \Big( 
           \osum_{\substack{
             t_C : C \to T\\
             t_A = \compi{\beta}{t_C}
           }} 
           \weight[\mathcal{T}]{t_C - (\beta \circ {-})} 
           \Big)
           \otimes
           \underbrace{
           \Big(
           \osum_{\substack{
             t_B : B \to T\\
             t_A = \compi{\alpha}{t_B}
           }} 
           \weight[\mathcal{T}]{t_B}
           \Big)
           }_{= \weight[\mathcal{T}]{\set{{-} \circ \alpha = t_A}}}
      \end{align*}
    Then the following holds: 
    \begin{enumerate}[label=(\Alph*)]
        \item\label{it:one:side:a}
            We have $\weight[\mathcal{T}]{D} = k$ if $\delta$ is weighable with $\mathcal{T}$.
        \item\label{it:one:side:b}
            We have $\weight[\mathcal{T}]{D} \le k$ if $\delta$ is bounded above by $\mathcal{T}$.
    \end{enumerate}
\end{lemma}

\begin{proof}%
  For part \ref{it:one:side:a} of the lemma we have:
  \begin{align*}
    \weight[\mathcal{T}]{D} 
    &= \osum_{t_D : D \to T} \weight[\mathcal{T}]{t_D} \\
    &= \osum_{t_A : A \to T} \; 
       \osum_{\substack{t_D : D \to T\\t_A = \compi{\compi{\alpha}{\beta'}}{t_D}}} 
       \weight[\mathcal{T}]{t_D} \\
    &= \osum_{t_A : A \to T} \; 
       \osum_{\substack{t_D : D \to T\\t_A = \compi{\compi{\alpha}{\beta'}}{t_D}}} 
       \weight[\mathcal{T}]{t_D \circ \beta'} \otimes \weight[\mathcal{T}]{t_D \circ \alpha' - (\beta \circ {-})} 
       &&\text{by Lemma~\ref{lem:decompose}}\\
    &= \osum_{t_A : A \to T} \; 
       \osum_{\substack{
         t_C : C \to T\\
         t_A = \compi{\beta}{t_C}
       }} \;
       \osum_{\substack{
         t_B : B \to T\\
         t_A = \compi{\alpha}{t_B}
       }} 
       \weight[\mathcal{T}]{t_B} \otimes \weight[\mathcal{T}]{t_C - (\beta \circ {-})} 
       &&\text{by Lemma~\ref{lem:bijection}}\\
    &\stackrel{(\dagger)}{=} \osum_{t_A : A \to T} \; 
       \osum_{\substack{
         t_C : C \to T\\
         t_A = \compi{\beta}{t_C}
       }} 
       \Big( \weight[\mathcal{T}]{t_C - (\beta \circ {-})} \otimes
       \osum_{\substack{
         t_B : B \to T\\
         t_A = \compi{\alpha}{t_B}
       }} 
       \weight[\mathcal{T}]{t_B} \Big) 
       \\
    &\stackrel{(\dagger)}{=} \osum_{t_A : A \to T} \; 
       \Big( 
       \osum_{\substack{
         t_C : C \to T\\
         t_A = \compi{\beta}{t_C}
       }} 
       \weight[\mathcal{T}]{t_C - (\beta \circ {-})} 
       \Big)
       \otimes
       \Big(
       \osum_{\substack{
         t_B : B \to T\\
         t_A = \compi{\alpha}{t_B}
       }} 
       \weight[\mathcal{T}]{t_B}
       \Big)
  \end{align*}
  The equalities marked with $(\dagger)$ are justified since
  $\weight[\mathcal{T}]{t_C - (\beta \circ {-})}$ and 
  $\osum_{\substack{
         t_B : B \to T\\
         t_A = \compi{\alpha}{t_B}
   }} \weight[\mathcal{T}]{t_B}$
  are constant.
  
  For part \ref{it:one:side:b} we get a similar derivation. The only difference concerns the step with the application of Lemma~\ref{lem:decompose} which yields $\le$ instead of $=$ under the assumptions for \ref{it:one:side:b}.
\end{proof}

\section{Termination via Weighted Type Graphs}
\label{sec:termination}

In the previous section, we have developed techniques for weighing pushout objects. Our next goal is to apply these techniques to the pushout squares of double pushout diagrams in order to prove that rewrite steps reduce the weight.

As we have already noted in Remark~\ref{remark:dpo:variants}, there exist many variations of the general concept of double pushout graph transformation. 
To keep our treatment as general as possible, we introduce in Section~\ref{sec:dpo:frameworks} an abstract concept of a double pushout framework $\aframework$ which maps rules $\arule: L \leftarrow K \to R$ to classes $\aframework(\arule)$ of double pushout diagrams with top span~$\arule$.

In order to prove that rewrite steps decrease the weight, we need to ensure that every object that can be rewritten admits some morphism into the weighted type graph.
To this end, we introduce the concept of context closures of rules in Section~\ref{sec:closures}.

In Section~\ref{sec:termination:proving}, we prove that rewrite steps are decreasing (Theorem~\ref{thm:steps}) if the rules satisfy certain criteria (Definition~\ref{def:decreasing:rule}).
Finally, we state our main theorem (Theorem~\ref{thm:termination}) for proving (relative) termination.

\subsection{Double Pushout Frameworks}
\label{sec:dpo:frameworks}

\begin{definition}
    A \emph{double pushout diagram} $\delta$ is a diagram of the form
      \begin{center}
        \begin{tikzpicture}[node distance=14mm]
          \node (I) {$K$};
          \node (L) [left of=I] {$L$};
          \node (R) [right of=I] {$R$};
          \node (G) [below of=L] {$G$};
          \node (C) [below of=I] {$C$};
          \node (H) [below of=R] {$H$};
          
          \draw [->] (I) to node [label] {$l$} (L);
          \draw [->] (I) to node [label] {$r$} (R);
          \draw [->] (L) to node [label] {$m$} (G);
          \draw [->] (I) to node [label] {$u$} (C);
          \draw [->] (R) to node [label] {$w$} (H);
          \draw [->] (C) to node [label] {$l'$} (G);
          \draw [->] (C) to node [label] {$r'$} (H);
    
          \node [at=($(I)!.5!(G)$)] {\normalfont PO};
          \node [at=($(I)!.5!(H)$)] {\normalfont PO};
        \end{tikzpicture}
      \end{center}
    This diagram $\delta$ is a \emph{witness} for the \emph{rewrite step} $G \step[\delta]{} H$.
    We use the following notation for the left and the right square of the diagram, respectively:
    \begin{align*}
      \leftsq{\delta} &= \posquare{K}{l}{L}{u}{C}{l'}{m}{G} \\
      \rightsq{\delta} &= \posquare{K}{r}{R}{u}{C}{r'}{w}{H}
    \end{align*}
    Both $\leftsq{\delta}$ and $\rightsq{\delta}$ are oriented pushout squares (where $u$ is vertical). 
    
    For a class $\Delta$ of double pushout diagrams, we define $\leftsq{\Delta} = \set{ \leftsq{\delta} \mid \delta \in \Delta }$ and 
    $\rightsq{\Delta} = \set{ \rightsq{\delta} \mid \delta \in \Delta }$.
\end{definition}

\begin{definition}[Double pushout framework]
    \label{def:double:pushout:framework}
    A \emph{double pushout framework $\aframework$} 
    is a mapping of DPO rules to classes of DPO diagrams such that, for every DPO rule $\rho$, $\aframework(\rho)$ is a class of DPO diagrams with top-span $\rho$.

    The rewrite relation $\step{\arule,\aframework}$ induced by a DPO rule $\rho$ in $\aframework$ is defined as follows:
    $G \step{\arule,\aframework} H$ iff $G \step[\delta]{} H$ for some $\delta \in \aframework(\arule)$.
    The rewrite relation~$\step{\asystem,\aframework}$ induced by a set~$\asystem$ of DPO rules is given by:
    $G \step{\asystem,\aframework} H$ iff $G \step{\arule,\aframework} H$ for some $\arule \in \asystem$.
    When $\aframework$ is clear from the context, we suppress $\aframework$ and write $\step{\arule}$ and $\step{\asystem}$.
\end{definition}

\subsection{Context Closures}
\label{sec:closures}%
\newcommand{\AR}{\mathcal{A}}%

For proving termination, we need to establish that every rewrite step causes a decrease in the weight of the host graph.
In particular, we need to ensure that every host graph admits some morphism into the weighted type graph.
For this purpose, we introduce the concept of a context closure. Intuitively, a context closure for a rule $L \leftto K \to R$ is a morphism $c: L \to T$ into the type graph $T$ in such a way that every match $L \to G$ can be mapped `around $c$'.

\begin{definition}[Context closure]
    Let $\AR$ be a class of morphisms.
    A \emph{context closure} of an object $L \in \ob{\categoryc}$ for class $\AR$ is an arrow $c: L \to T$ such that, for every arrow $\alpha: L \to G$ in $\AR$, there exists a morphism $\beta: G \to T$ with $c = \beta \circ \alpha$. This can be depicted as 
    \begin{center}
    \begin{tikzpicture}[baseline=(X.base),node distance=11mm,nodes={rectangle,inner sep=1mm,outer sep=0}]
        \node at (0,-4mm) {$=$};
        \node (Z) at (11mm,0mm) {$G$};
        \node (Y) at (0,-9mm) {$T$};
        \node (X) at (-11mm,0mm) {$L$};
        \draw [->] (Z) to node [label,pos=0.4] {$\beta$} (Y);
        \draw [->] (X) to node [label,pos=0.5] {$\alpha$} (Z); 
        \draw [->] (X) to node [label,pos=0.4] {$c$} (Y); 
    \end{tikzpicture}
    \end{center}
\end{definition}

\newcommand{\MM}{\mathcal{M}}

\begin{remark}[Partial map classifiers as context closures]
    The notion of a context closure is strictly weaker than a partial map classifier.
    If category $\categoryc$ has an $\MM$-partial map classifier $(T, \eta)$~(see \cite[Definition 28.1]{adamek2009joy} and \cite[Definition 5]{corradini2020algebraic}) for a stable class of monics $\MM$, then the $\MM$-partial map classifier arrow $\eta_X : X \hookrightarrow T(X)$ for an object $X$ is a context closure of $X$ for $\MM$, because for any $\alpha: X \hookrightarrow Z$ in $\MM$ there exists a unique $\beta: Z \to T(X)$ making the following square a pullback:
    \begin{center}
    \begin{tikzpicture}[baseline=(X.base),node distance=14mm,nodes={rectangle,inner sep=0.5mm,outer sep=0}]
        \node (X) {$X$};
        \node (Y) [below of=X] {$T(X)$};
        \node (X') [right of=X,node distance=14mm] {$X$};
        \node (Z) [below of=X'] {$Z$};
        \draw [regmono] (X) to node [label,pos=0.4] {$\eta_X$} (Y); 
        \draw [double,double distance=.5mm] (X) to (X'); 
        \draw [->,densely dotted] (Z) to node [label,pos=0.4] {$\beta$} (Y);
        \draw [regmono] (X') to node [label,pos=0.4] {$\alpha$} (Z); 
        \node at ($(X)!.5!(Z)$) {PB};
    \end{tikzpicture}
    \end{center}
\end{remark}

\begin{definition}[Context closure for a rule]\label{def:context:closure}
    Let $\aframework$ be a DPO framework and $\arule$ a DPO rule $L \stackrel{}{\leftarrow} K \stackrel{}{\to} R$.
    A \emph{context closure for $\arule$ and $T$ (in $\aframework$)} is a context closure $c : L \to T$ for the class of match morphisms occurring in $\aframework(\arule)$.
\end{definition}

\begin{remark}[Context closures for unrestricted matching]
    \label{lem:flower:unrestricted}%
    Consider a DPO rule $\arule \; = \; {L \leftarrow K \to R}$ in an arbitrary DPO framework $\aframework$ (in particular, matching can be unrestricted).
    Every morphism $c: L \to 1 \to T$, that factors through the terminal object $1$, is a context closure for $\rho$ and $T$ in $\aframework$.

    This particular choice of a context closure generalizes the flower nodes from~\cite{bruggink2015} which are defined specifically for the category of edge-labelled multigraphs.
    A flower node~\cite{bruggink2015} in a graph $T$ is just a morphism $1 \to T$.
    
    When employing such context closures, we establish termination with respect to unrestricted matching. This is bound to fail if termination depends on the matching restrictions.
\end{remark}

\begin{remark}[Increasing the power for monic matching]
    \label{rem:power:monic}
    For rules $L \stackrel{l}{\leftarrow} K \stackrel{r}{\to} R$ with monic matching, it is typically fruitful to choose a type graph $T$ and context closure $c: L \to T$ in such a way that $c$ avoids collapsing $l(K)$ in~$L$. In other words, the context closure $c$ should be monic for $\set{l \circ h \mid h \in \homset{-}{K}}$. In the case that $l$ is monic, this is equivalent to $c \circ l$ being monic.
    
    For \Graph{}, the intuition is as follows. When collapsing nodes of $l(K)$ along~$c$, we actually prove a stronger termination property than intended. We then prove termination even if matches $m : L \to G$ are allowed to collapse those nodes.
\end{remark}

\begin{example}[Context closure for monic matching]\label{ex:monic:closure}
    Let $\arule$ be the DPO rule
    \begin{center}
        \scalebox{1}{{%
\newcommand{\nodex}{\vertex{x}{cblue!20}}%
\newcommand{\nodey}{\vertex{y}{cgreen!20}}%
\newcommand{\nodez}{\vertex{z}{corange!20}}%
\newcommand{\nodew}{\vertex{w}{cred!20}}%
\newcommand{\nodectx}{\vertex{c}{white}}%
\begin{tikzpicture}[baseline=-10mm,->,node distance=12mm,n/.style={},epattern/.style={very thick},mred/.style={cred,very thick},mblue/.style={cblue,very thick},mgreen/.style={cdarkgreen,very thick},mpurple/.style={cpurple,very thick,densely dotted}]
        \graphbox{$L$}{0mm}{0mm}{30mm}{15mm}{0mm}{-3.5mm}{
          \node [npattern] (x) at (-6mm,-8mm) {\nodex};
          \node [npattern] (y) at (6mm,-8mm) {\nodey};
          \node [npattern] (z) at (6mm,0mm) {\nodez};
          \draw [->] (x) to (y);
          \draw [->] (y) to[bend right=20] (z);
          \draw [->] (z) to[bend right=20] (y);
        }
        \graphbox{$K$}{38mm}{0mm}{30mm}{15mm}{0mm}{-3.5mm}{
          \node [npattern] (x) at (-6mm,-8mm) {\nodex};
          \node [npattern] (y) at (6mm,-8mm) {\nodey};
        }
        \graphbox{$R$}{76mm}{0mm}{30mm}{15mm}{0mm}{-3.5mm}{
          \node [npattern] (x) at (-6mm,-8mm) {\nodex};
          \node [npattern] (y) at (6mm,-8mm) {\nodey};
          \node [npattern] (z) at (0mm,0mm) {\nodew};
          \draw [->] (x) to (y);
          \draw [->] (y) to[bend left=0] (z);
          \draw [->] (z) to[bend left=0] (x);
        }
        \begin{scope}[nodes={fill=white,label,inner sep=0.5mm}]
        \draw [->] (37mm,-7.5mm) to node {$l$} ++(-6mm,0mm);
        \draw [->] (69mm,-7.5mm) to node {$r$} ++(6mm,0mm);
        \end{scope}
\end{tikzpicture}
}%
}
    \end{center}
    from~\cite[Example~6.1]{overbeek2024termination} in \Graph with a framework $\aframework$ with monic matching.

    Following Remark~\ref{rem:power:monic}, we construct a context closure $c : L \to T$ for $\arule$ that avoids collapsing $l(K)$ of $L$.
    Since $l(K) \iso K$, this means that the complete graph around the nodes of $K$ must be a subobject of $T$.
    We pick
    \begin{center}
        \scalebox{1}{{%
\newcommand{\nodex}{\vertex{x}{cblue!20}}%
\newcommand{\nodey}{\vertex{y}{cgreen!20}}%
\newcommand{\nodez}{\vertex{z}{corange!20}}%
\newcommand{\nodew}{\vertex{w}{cred!20}}%
\newcommand{\nodeu}{\vertex{u}{white}}%
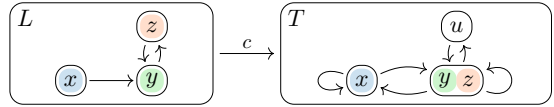
\begin{tikzpicture}[baseline=-10mm,->,node distance=12mm,n/.style={},epattern/.style={very thick},mred/.style={cred,very thick},mblue/.style={cblue,very thick},mgreen/.style={cdarkgreen,very thick},mpurple/.style={cpurple,very thick,densely dotted}]
        \graphbox{$L$}{0mm}{0mm}{30mm}{15mm}{0mm}{-3.5mm}{
          \node [npattern] (x) at (-6mm,-8mm) {\nodex};
          \node [npattern] (y) at (6mm,-8mm) {\nodey};
          \node [npattern] (z) at (6mm,0mm) {\nodez};
          \draw [->] (x) to (y);
          \draw [->] (y) to[bend right=20] (z);
          \draw [->] (z) to[bend right=20] (y);
        }
        \graphbox{$T$}{40mm}{0mm}{40mm}{15mm}{-2mm}{-3.5mm}{
          \node [npattern] (x) at (-6mm,-8mm) {\nodex};
          \node [npattern] (y) at (8mm,-8mm) {\nodey\nodez};
          \node [npattern] (z) at (8mm,0mm) {\nodeu};
          \draw [->] (x) to[bend left=20] (y);
          \draw [->] (y) to[bend left=20] (x);
          \draw [->] (x) to[thinloop=180,distance=1.5em] (x);
          \draw [->] (y) to[thinloop=0,distance=1.5em] (y);
          \draw [->] (y) to[bend right=20] (z);
          \draw [->] (z) to[bend right=20] (y);
        }
        \begin{scope}[nodes={fill=white,label,inner sep=0.5mm}]
        \draw [->] (31mm,-7.5mm) to node [] {$c$} ++(8mm,0mm);
        \end{scope}
\end{tikzpicture}
}%
}
    \end{center}
    Indeed, for every monic match $m : L \mono G$ there exists a morphism $\alpha : G \to T$ such that $c = \alpha \circ m$. 
    Thus $c$ is a context closure for $\arule$ and $T$ in $\aframework$.
\end{example}

\subsection{Proving Termination}
\label{sec:termination:proving}

In this section we prove our main theorem (Theorem~\ref{thm:termination}) for using weighted type graphs to prove termination of graph transformation systems. 

First, we introduce different notions of decreasing rules (Defintion~\ref{def:decreasing:rule}). Then we prove that rewrite steps arising from these rules cause a decrease in the weight of the host graphs (Theorem~\ref{thm:steps}). Finally, we summarize our technique for proving (relative) termination (Theorem~\ref{thm:termination}).

The following definition introduces the notions of weakly, uniformly and closure decreasing rules. Weakly decreasing rules never increase the weight of the host graph. 
Uniformly decreasing rules reduce the weight of the host graph for every mapping of the host graph onto the type graph. 
If the underlying semiring is strictly monotonic, then closure decreasingness suffices: the weight of the host graph decreases for those mappings that correspond to the context closure of the rule (and is non-increasing for the other mappings).

\begin{definition}[Decreasing rules]\label{def:decreasing:rule}
    Let $\aframework$ be a DPO framework.
    Moreover, let $\awtg = \wtg$ be a finitary weighted type graph over a \wf{} commutative semiring $\wfring$.
    A DPO rule $\arule: L \stackrel{l}{\leftarrow} K \stackrel{r}{\to} R$ is called
    \begin{enumerate}[label=(\roman*)]
        \item 
            \emph{weakly decreasing} (w.r.t.\ $\awtg$ in $\aframework$) if
            \begin{itemize}
                \item 
                    $\weighti[\mathcal{T}]{t_K}{l} \oweaki \weighti[\mathcal{T}]{t_K}{r}$ for every $t_K : K \to T$;
            \end{itemize}
        \item 
            \emph{uniformly decreasing} (w.r.t.\ $\awtg$ in $\aframework$) if
            \begin{itemize}
                \item 
                    for every $t_K : K \to T$:\\
                    $ \weighti[\mathcal{T}]{t_K}{l} \ostricti \weighti[\mathcal{T}]{t_K}{r}$ or\\
                    $\set{{-} \circ l = t_K} = \emptyset = \set{{-} \circ r = t_K}$, and
                \item 
                    there exists a context closure $\flower_\arule$ for $\arule$ and $T$ in $\aframework$;
            \end{itemize}
    \end{enumerate}
    If semiring $S$ is moreover strictly monotonic, we say that $\arule$ is
    \begin{enumerate}[label=(\roman*),resume]
        \item 
            \emph{closure decreasing} (w.r.t.\ $\awtg$ in $\aframework$) if
            \begin{itemize}
                \item 
                    $\arule$ is weakly decreasing, 
                \item 
                    there exists a context closure $\flower_\arule$ for $\arule$ and $T$ in $\aframework$, and 
                \item 
                    $\weighti[\mathcal{T}]{t_K}{l} \ostricti \weighti[\mathcal{T}]{t_K}{r}$ for $t_K = \flower_\arule \circ l$.
            \end{itemize}
    \end{enumerate}
\end{definition}

By saying that a rule is closure decreasing, we tacitly assume that the semiring is strictly monotonic (for otherwise the concept is not defined). Note that if $S$ is strictly monotonic, uniformly decreasing implies closure decreasing.

\begin{remark}\label{rem:strongly:decreasing}
    Our notions of uniformly and closure decreasing rules generalize the corresponding concepts (strongly and strictly decreasing rules) in~\cite{bruggink2014,bruggink2015}.
    
    In~\cite{bruggink2015}, the notion of `strongly decreasing' rules requires that
    \begin{align*}
      \weighti[\mathcal{T}]{t_K}{l} \ostricti \weighti[\mathcal{T}]{t_K}{r}  
    \end{align*}
    for every $t_K : K \to T$.
    Bruggink et al.~\cite{bruggink2015} claim that their notion generalizes the termination techniques from their earlier work~\cite{bruggink2014} focused on the tropical and arctic semirings.
    However, this is not the case since there typically exist morphisms $t_K : K \to T$ for which $\set{{-} \circ l = t_K} = \emptyset$ and thus $\weighti[\mathcal{T}]{t_K}{l} = \ozero$:
    \begin{enumerate}
    \item 
        In the arctic semiring $\ozero = -\infty$ and $-\infty \ostricti \weighti[\mathcal{T}]{t_K}{r}$ is impossible. In contrast, the paper~\cite{bruggink2014} allows for $\weighti[\mathcal{T}]{t_K}{r} = -\infty$ in this case.
    \item
        In the tropical semiring $\ozero = \infty$ and $\infty \ostricti \weighti[\mathcal{T}]{t_K}{r}$ enforces that $\set{{-} \circ r = t_K} \ne \emptyset$ while~\cite{bruggink2014} allows for $\weighti[\mathcal{T}]{t_K}{r} = \infty$ in this case.
    \end{enumerate}
    Our notion of uniform decreasingness properly generalizes the techniques in~\cite{bruggink2014} and works for any well-founded commutative semiring.  
\end{remark}

\newcommand{\weightC}{w_{C,t_K}}

\begin{theorem}[Decreasing steps]\label{thm:steps}
    Let $\aframework$ be a DPO framework.
    Moreover, let $\awtg = \wtg$ be a finitary weighted type graph over a \wf{} commutative semiring $\wfring$.
    Let $\arule$ be a DPO rule and $\delta \in \aframework(\arule)$ such that
    \begin{itemize}
        \item $\leftsq{\delta}$ is weighable with $\mathcal{T}$, and
        \item $\rightsq{\delta}$ is bounded above by $\mathcal{T}$.
    \end{itemize}
    Then, for the induced rewrite step $G \step[\delta]{} H$, we have
    \begin{enumerate}
        \item $\weight[\mathcal{T}]{G} \oweaki \weight[\mathcal{T}]{H}$ if $\arule$ is weakly decreasing;
        \item $\weight[\mathcal{T}]{G} \ostricti \weight[\mathcal{T}]{H}$ if $\arule$ is uniformly or closure decreasing.
    \end{enumerate}
\end{theorem}

\begin{proof}[Proof of Theorem~\ref{thm:steps}]
    Let $\delta$ be the following double pushout diagram:
      \begin{center}
        \begin{tikzpicture}[node distance=14mm]
          \node (I) {$K$};
          \node (L) [left of=I] {$L$};
          \node (R) [right of=I] {$R$};
          \begin{scope}          
          \node (G) [below of=L] {$G$};
          \node (C) [below of=I] {$C$};
          \node (H) [below of=R] {$H$};
          \end{scope}
          
          \draw [->] (I) to node [label] {$l$} (L);
          \draw [->] (I) to node [label] {$r$} (R);
          \draw [->] (L) to node [label,pos=0.4] {$m$} (G);
          \draw [->] (I) to node [label,pos=0.4] {$u$} (C);
          \draw [->] (R) to node [label,pos=0.4] {$w$} (H);
          \draw [->] (C) to node [label] {$l'$} (G);
          \draw [->] (C) to node [label] {$r'$} (H);
    
          \node [at=($(I)!.5!(G)$)] {\normalfont PO};
          \node [at=($(I)!.5!(H)$)] {\normalfont PO};
        \end{tikzpicture}
      \end{center}
    For $t_K : K \to T$, define
    \begin{align*}
       \weightC{} = %
       \osum_{\substack{
         t_C : C \to T\\
         t_K = \compi{u}{t_C}
       }} 
       \weight[\mathcal{T}]{t_C - (u \circ -)} 
    \end{align*}
    By Lemma~\ref{lem:weight:one}, $\oone \oweak \weight[\mathcal{T}]{t_C - (u \circ -)} \ne \ozero$.
    By Proposition~\ref{prop:wf:semiring}, for all $x,y \in S$:
    \begin{enumerate}[label=(\alph*)]
        \item \label{it:weightC:weak} $x \oweak y \implies \weightC{} \otimes x \oweak \weightC{} \otimes y$;
        \item \label{it:weightC:strict} 
        $x \ostrict y \implies \weightC{} \otimes x \ostrict \weightC{} \otimes y$
        if there exists $t_C$ with $t_K = \compi{u}{t_C}$;
        \item \label{it:weightC:zero} 
        $\weightC{} \otimes x = \ozero = \weightC{} \otimes y$
        if there does not exist $t_C$ with $t_K = \compi{u}{t_C}$;
        \item \label{it:weightC:flower}
        If there is a context closure $\flower_\arule$ for $\arule$ and $T$ in $\aframework$, then
        \begin{align*}
        x \ostrict y \implies \weightC{} \otimes x \ostrict \weightC{} \otimes y
        \end{align*}
        for $t_K = \flower_\arule \circ l$.
        This can be argued as follows. By the definition of context closure, there exists $\alpha : G \to T$ such that $\flower_\arule = \alpha \circ m$.  
        Define $t_C = \alpha \circ l'$.
        Then $t_K = \alpha \circ m \circ l = \alpha \circ l' \circ u = t_C \circ u$ and the claim follows from \ref{it:weightC:strict}.
    \end{enumerate}
    Since $\leftsq{\delta}$ is weighable with $\mathcal{T}$, by Lemma~\ref{lem:one:side} we obtain
      \begin{align*}
        \weight[\mathcal{T}]{G}
        &= \osum_{t_K : K \to T} \; 
           \weightC{} \otimes
           \weight[\mathcal{T}]{\set{{-} \circ l = t_K}}
      \end{align*}
    Since $\rightsq{\delta}$ is bounded above by $\mathcal{T}$, by Lemma~\ref{lem:one:side} we obtain
      \begin{align*}
        \weight[\mathcal{T}]{H}
        &\le \osum_{t_K : K \to T} \; 
           \weightC{} \otimes
           \weight[\mathcal{T}]{\set{{-} \circ r = t_K}}
      \end{align*}
    We distinguish cases:
    \begin{enumerate}
        \item 
            If $\arule$ is  weakly decreasing, then 
            $\weighti[\mathcal{T}]{t_K}{l} \oweaki \weighti[\mathcal{T}]{t_K}{r}$ for every $t_K : K \to T$.
            By~\ref{it:weightC:weak} we have: 
            \begin{align}
              \weightC \otimes \weighti[\mathcal{T}]{t_K}{l} \oweaki \weightC \otimes \weighti[\mathcal{T}]{t_K}{r} 
              \label{steps:weightC:ge}
            \end{align}
            for every $t_K : K \to T$.
            Hence $\weight[\mathcal{T}]{G} \oweaki  \weight[\mathcal{T}]{H}$ by~\eqref{wfring:oplus:oweak}.
        \item
            If $\arule$ is uniformly decreasing, then from~\ref{it:weightC:strict} and~\ref{it:weightC:zero} we obtain
            \begin{enumerate}[label=(\roman*)]
                \item \label{it:strict} $\weightC \otimes \weighti[\mathcal{T}]{t_K}{l} \ostricti \weightC \otimes \weighti[\mathcal{T}]{t_K}{r}$, or
                \item $\weightC \otimes \weighti[\mathcal{T}]{t_K}{l} = \ozero = \weightC \otimes \weighti[\mathcal{T}]{t_K}{r}$.
            \end{enumerate}
            for every $t_K : K \to T$.
            To establish $\weight[\mathcal{T}]{G} \ostricti \weight[\mathcal{T}]{H}$, using~\eqref{wfring:oplus:ostrict}, 
            it suffices to show that we have case~\ref{it:strict} for some $t_K : K \to T$.
            This follows from \ref{it:weightC:flower} since we have a context closure $\flower_\arule$ for $\arule$ and $T$ by assumption.
        \item
            If $\arule$ is closure decreasing, 
            then it is also weakly decreasing and we obtain \eqref{steps:weightC:ge} for every $t_K : K \to T$.
            Since the semiring is strictly monotonic,  by~\eqref{wfring:strict}, it suffices to show that there exists some $t_K : K \to T$ such that
            \begin{align}
              \weightC \otimes \weighti[\mathcal{T}]{t_K}{l} \ostricti \weightC \otimes \weighti[\mathcal{T}]{t_K}{r} 
              \label{steps:weightC:gt}
            \end{align}
            in order to conclude $\weight[\mathcal{T}]{G} \ostricti \weight[\mathcal{T}]{H}$.
            There is a context closure $\flower_\arule$ for $\arule$ and $T$. 
            By assumption $\weighti[\mathcal{T}]{t_K}{l} \ostricti \weighti[\mathcal{T}]{t_K}{r}$ for $t_K = \flower_\arule \circ l$,
            and we obtain~\eqref{steps:weightC:gt} by~\ref{it:weightC:flower}. \qedhere
    \end{enumerate}
\end{proof}

We now state our main theorem of proving (relative) termination using weighted type graphs.
\begin{theorem}[Proving termination]\label{thm:termination}
    Let $\aframework$ be a DPO framework and let $\asystem_1$ and $\asystem_2$ be sets of double pushout rules.
    Let $\wfring$ be a \wf{} commutative semiring, and 
    let $\awtg = \wtg$ be a finitary weighted type graph such that
    \begin{itemize}
        \item $\leftsq{\delta}$ is weighable with $\mathcal{T}$, and
        \item $\rightsq{\delta}$ is bounded above by $\mathcal{T}$
    \end{itemize}
    for every rule $\arule \in (\asystem_1 \cup \asystem_2)$ and double pushout diagram $\delta \in \aframework(\arule)$. If
    \begin{enumerate}
        \item 
        $\arule$ is weakly decreasing for every $\arule \in \asystem_2$, and
        \item 
        $\arule$ is uniformly or closure decreasing for every $\arule \in \asystem_1$,
    \end{enumerate}
    then $\asystem_1$ is terminating relative to $\asystem_2$.
\end{theorem}

\begin{proof}[Proof of Theorem~\ref{thm:termination}]
    Follows immediately from Theorem~\ref{thm:steps}.
\end{proof}

\section{Examples}
\label{sec:examples}

\begin{notation}[Visual notation]
\label{notation:visual}
    We use the visual notation as in~\cite{overbeek2024termination}.
    The vertices of graphs are non-empty sets $\{ x_1,\ldots,x_n \}$ depicted by boxes
    {%
    \hspace{-3.4mm}
    \newcommand{\nodexa}{\vertex{x_1}{cblue!20}}
    \newcommand{\nodexb}{\vertex{x_n}{cblue!20}}
    \begin{tikzpicture}[->,node distance=12mm,n/.style={},baseline=-0.8ex]
    \node [npattern] (xaxb)
          {\nodexa \ \raisebox{1mm}{$\cdots$} \  \nodexb};
    \end{tikzpicture}%
    }.
    We will choose the vertices of the graphs in such a way that the homomorphisms $h : G \to H$ are fully determined by $S \subseteq h(S)$ for all $S \in V_G$. For instance, in Example~\ref{ex:monic:triangle} below, the morphism $\flower_\arule: L \to T$ maps nodes $\{ x \}, \{ y \}, \{ z \} \in L$ as follows: $\flower_\arule(\{x\}) = \{ x \}$ and $\flower_\arule(\{y\}) = \flower_\arule(\{z\}) = \{ y,z \} \in T$.
\end{notation}

We start with two examples (Examples~\ref{ex:monic:loop} and~\ref{ex:monic:triangle}) for which the approaches from~\cite{bruggink2014,bruggink2015} fail, as the systems are not terminating with unrestricted matching.

\begin{example}[Loop unfolding with monic matching]
\label{ex:monic:loop}
    Let $\arule$ be the DPO rule 
    \begin{center}
        \scalebox{1}{{%
\newcommand{\nodex}{\vertex{x}{cblue!20}}%
\newcommand{\nodey}{\vertex{y}{cgreen!20}}%
\newcommand{\nodez}{\vertex{z}{corange!20}}%
\newcommand{\nodew}{\vertex{w}{cred!20}}%
\newcommand{\nodectx}{\vertex{c}{white}}%
\begin{tikzpicture}[baseline=-5mm,->,node distance=12mm,n/.style={},epattern/.style={very thick},mred/.style={cred,very thick},mblue/.style={cblue,very thick},mgreen/.style={cdarkgreen,very thick},mpurple/.style={cpurple,very thick,densely dotted}]
        \graphbox{$L$}{0mm}{0mm}{28mm}{7mm}{1mm}{-3.5mm}{
          \node [npattern] (x) at (-6mm,0mm) {\nodex};
          \node [npattern] (y) at (6mm,0mm) {\nodey};
          \draw (x) to[thinloop=0,distance=1.5em] (x);
        }
        \graphbox{$K$}{36mm}{0mm}{28mm}{7mm}{1mm}{-3.5mm}{
          \node [npattern] (x) at (-6mm,0mm) {\nodex};
          \node [npattern] (y) at (6mm,0mm) {\nodey};
        }
        \begin{scope}[opacity=1]        
        \graphbox{$R$}{72mm}{0mm}{28mm}{7mm}{1mm}{-3.5mm}{
          \node [npattern] (x) at (-6mm,0mm) {\nodex};
          \node [npattern] (y) at (6mm,0mm) {\nodey};
          \draw (x) to (y);
        }
        \end{scope}
        \begin{scope}[nodes={fill=white,label,inner sep=0.5mm}]
        \draw [->] (35mm,-3.5mm) to node {$l$} ++(-6mm,0mm);
        \draw [->] (65mm,-3.5mm) to node {$r$} ++(6mm,0mm);
        \end{scope}
    \end{tikzpicture}
}%
}
    \end{center}
    in \Graph, and let matches $m$ restrict to monic matches.
    We prove termination using the weighted type graph
    \begin{center}
      \scalebox{1}{{%
\newcommand{\nodex}{\vertex{x}{cblue!20}}%
\newcommand{\nodey}{\vertex{y}{cgreen!20}}%
\newcommand{\nodez}{\vertex{z}{corange!20}}%
\newcommand{\nodew}{\vertex{w}{cred!20}}%
\newcommand{\nodeu}{\vertex{u}{white}}%
\begin{tikzpicture}[baseline=(x.base),->,node distance=12mm,n/.style={},epattern/.style={very thick},mred/.style={cred,very thick},mblue/.style={cblue,very thick},mgreen/.style={cdarkgreen,very thick},mpurple/.style={cpurple,very thick,densely dotted}]
        \graphbox{$\awtg$}{84mm}{0mm}{40mm}{7mm}{2mm}{6.5mm}{
          \node [npattern] (x) at (-6mm,-10mm) {\nodex};
          \node [npattern] (y) at (6mm,-10mm) {\nodey};
          \draw [->] (x) to[bend left=20] (y);
          \draw [->] (y) to[bend left=20] (x);
          \draw [->] (x) to[thinloop=180,distance=1.5em] node [left,label] {$2$} (x);
          \draw [->] (y) to[thinloop=0,distance=1.5em] node [right,label] {$2$} (y);
        }
\end{tikzpicture}
}%
}
    \end{center}
    over the (strictly monotonic) arithmetic semiring.
    The type graph $T$ is unlabelled.
    The numbers along the edges indicate the weighted elements $\elements$. 
    Here $\elements$ consists of $2$ morphisms with domain $\edgeGraph{}$ (see further Remark~\ref{rem:traceable:graph}) and their weight is the number along the edge they target.
    So $\elements = \set{e_x,e_y}$ where $e_x$ and $e_y$ map $\edgeGraph{}$ onto the loop on $\{x\}$ and the loop on $\{y\}$ in $T$, respectively. The weights of $e_x$ and $e_y$ are $\ww(e_x) = \ww(e_y) = 2$. 
        
    As the context closure $\flower_\arule$ for $\arule$ and $T$ we use the inclusion $\flower_\arule: L \mono T$. 
    Clearly, for any monic match $L \mono G$, there exists a map $\alpha: G \to T$ such that $\flower_\arule = \alpha \circ m$.
    For $t_K = \flower_\arule \circ l$ we have 
      $\weighti[\mathcal{T}]{t_K}{l} = {2} \ostricti {1} = \weighti[\mathcal{T}]{t_K}{r}$  
    since the empty product is $1$.
    Moreover, $\weighti[\mathcal{T}]{t_K}{l} \oweaki \weighti[\mathcal{T}]{t_K}{r}$ for all other $t_K : K \to T$.
    Thus $\arule$ is closure decreasing, and consequently terminating by Theorem~\ref{thm:termination}.
\end{example}

\begin{example}[Reconfiguration]
\label{ex:monic:triangle}
    Continuing Example~\ref{ex:monic:closure},
    let $\arule$ be the DPO rule 
    \begin{center}
        \scalebox{1}{{%
\newcommand{\nodex}{\vertex{x}{cblue!20}}%
\newcommand{\nodey}{\vertex{y}{cgreen!20}}%
\newcommand{\nodez}{\vertex{z}{corange!20}}%
\newcommand{\nodew}{\vertex{w}{cred!20}}%
\newcommand{\nodectx}{\vertex{c}{white}}%
\begin{tikzpicture}[baseline=-10mm,->,node distance=12mm,n/.style={},epattern/.style={very thick},mred/.style={cred,very thick},mblue/.style={cblue,very thick},mgreen/.style={cdarkgreen,very thick},mpurple/.style={cpurple,very thick,densely dotted}]
        \graphbox{$L$}{0mm}{0mm}{30mm}{15mm}{0mm}{-3.5mm}{
          \node [npattern] (x) at (-6mm,-8mm) {\nodex};
          \node [npattern] (y) at (6mm,-8mm) {\nodey};
          \node [npattern] (z) at (6mm,0mm) {\nodez};
          \draw [->] (x) to (y);
          \draw [->] (y) to[bend right=20] (z);
          \draw [->] (z) to[bend right=20] (y);
        }
        \graphbox{$K$}{38mm}{0mm}{30mm}{15mm}{0mm}{-3.5mm}{
          \node [npattern] (x) at (-6mm,-8mm) {\nodex};
          \node [npattern] (y) at (6mm,-8mm) {\nodey};
        }
        \graphbox{$R$}{76mm}{0mm}{30mm}{15mm}{0mm}{-3.5mm}{
          \node [npattern] (x) at (-6mm,-8mm) {\nodex};
          \node [npattern] (y) at (6mm,-8mm) {\nodey};
          \node [npattern] (z) at (0mm,0mm) {\nodew};
          \draw [->] (x) to (y);
          \draw [->] (y) to[bend left=0] (z);
          \draw [->] (z) to[bend left=0] (x);
        }
        \begin{scope}[nodes={fill=white,label,inner sep=0.5mm}]
        \draw [->] (37mm,-7.5mm) to node {$l$} ++(-6mm,0mm);
        \draw [->] (69mm,-7.5mm) to node {$r$} ++(6mm,0mm);
        \end{scope}
\end{tikzpicture}
}%
}
    \end{center}
    from~\cite[Example~6.1]{overbeek2024termination}
    in \Graph, and let matches $m$ restrict to monic matches.
    We prove termination of $\arule$ with matches $m$ restrict to monic matches. We use the weighted type graph 
    \begin{center}
        \scalebox{1}{{%
\newcommand{\nodex}{\vertex{x}{cblue!20}}%
\newcommand{\nodey}{\vertex{y}{cgreen!20}}%
\newcommand{\nodez}{\vertex{z}{corange!20}}%
\newcommand{\nodew}{\vertex{w}{cred!20}}%
\newcommand{\nodeu}{\vertex{u}{white}}%
\begin{tikzpicture}[baseline=-10mm,->,node distance=12mm,n/.style={},epattern/.style={very thick},mred/.style={cred,very thick},mblue/.style={cblue,very thick},mgreen/.style={cdarkgreen,very thick},mpurple/.style={cpurple,very thick,densely dotted}]
        \graphbox{$\awtg$}{84mm}{0mm}{40mm}{16mm}{-1mm}{-3.5mm}{
          \node [npattern] (x) at (-6mm,-9mm) {\nodex};
          \node [npattern] (y) at (8mm,-9mm) {\nodey\nodez};
          \node [npattern] (z) at (8mm,0mm) {\nodeu};
          \draw [->] (x) to[bend left=20] (y);
          \draw [->] (y) to[bend left=20] (x);
          \draw [->] (x) to[thinloop=180,distance=1.5em] (x);
          \draw [->] (y) to[thinloop=0,distance=1.5em] (y);
          \draw [->] (y) to[bend right=20] node [right,label] {$2$} (z);
          \draw [->] (z) to[bend right=20] (y);
        }
\end{tikzpicture}
}%
}
    \end{center}
    over the (strictly monotonic) arithmetic semiring.
    This is the type graph constructed in Example~\ref{ex:monic:closure}, enriched with weighted elements.
    $\elements$ consists of a single morphism with weight $2$, mapping $\edgeGraph{}$ onto the edge from $\{y,z\}$ to $\{u\}$ in $T$.

    As the context closure $\flower_\arule$ for $\arule$ and $T$ we use $\flower_\arule: L \to T$ as indicated by the node names ($\{y\}$ and $\{z\}$ are mapped onto $\{y,z\}$). 
    For $t_K = \flower_\arule \circ l$ we then obtain 
      $\weighti[\mathcal{T}]{t_K}{l} = {1 + 1 + 2} \ostricti {1 + 1} = \weighti[\mathcal{T}]{t_K}{r}$  
    since the empty product is $1$, and $\{z\}$ in $L$ can be mapped onto $\{x\}$, $\{y,z\}$ and $\{u\}$,
    while $\{w\}$ in $R$ can only be mapped onto $\{x\}$ and $\{y,z\}$ in $T$.
    Furthermore, $\weighti[\mathcal{T}]{t_K}{l} \oweaki \weighti[\mathcal{T}]{t_K}{r}$ for all other $t_K : K \to T$.
    Thus $\arule$ is closure decreasing, and hence terminating by Theorem~\ref{thm:termination}.
\end{example}

We continue with an example (Example~\ref{ex:monic:simple}) of rewriting simple graphs. Here the techniques from~\cite{bruggink2014,bruggink2015} are not applicable, because they are defined only for multigraphs.
Even when considering the rule as a rewrite rule on multigraphs, the technique from~\cite{bruggink2014} is not applicable due to non-monic $r$, and the technique from~\cite{bruggink2015} fails because the rule is not terminating with respect to unrestricted matching.

\begin{example}[Simple graphs and monic matching]
\label{ex:monic:simple}
    Let $\arule$ be the DPO rule
    \begin{center}
        \scalebox{1}{{%
\newcommand{\nodex}{\vertex{x}{cblue!20}}%
\newcommand{\nodey}{\vertex{y}{cgreen!20}}%
\newcommand{\nodez}{\vertex{z}{corange!20}}%
\newcommand{\nodew}{\vertex{w}{cred!20}}%
\newcommand{\nodectx}{\vertex{c}{white}}%
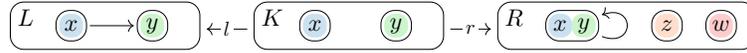
\begin{tikzpicture}[baseline=-5mm,->,node distance=12mm,n/.style={},epattern/.style={very thick},mred/.style={cred,very thick},mblue/.style={cblue,very thick},mgreen/.style={cdarkgreen,very thick},mpurple/.style={cpurple,very thick,densely dotted}]
        \graphbox{$L$}{0mm}{0mm}{28mm}{7mm}{1mm}{-3.5mm}{
          \node [npattern] (x) at (-6mm,0mm) {\nodex};
          \node [npattern] (y) at (6mm,0mm) {\nodey};
          \draw (x) to (y);
        }
        \graphbox{$K$}{36mm}{0mm}{28mm}{7mm}{1mm}{-3.5mm}{
          \node [npattern] (x) at (-6mm,0mm) {\nodex};
          \node [npattern] (y) at (6mm,0mm) {\nodey};
        }
        \begin{scope}[opacity=1]        
        \graphbox{$R$}{72mm}{0mm}{38mm}{7mm}{-8mm}{-3.5mm}{
          \node [npattern] (x) at (0mm,0mm) {\nodex\nodey};
          \node [npattern] (z) at (14mm,0mm) {\nodez};
          \node [npattern] (w) at (22mm,0mm) {\nodew};
          \draw (x) to [thinloop=0,distance=1.5em] (x);
        }
        \end{scope}
        \begin{scope}[nodes={fill=white,label,inner sep=0.5mm}]
        \draw [->] (35mm,-4mm) to node {$l$} ++(-6mm,0mm);
        \draw [->] (65mm,-4mm) to node {$r$} ++(6mm,0mm);
        \end{scope}
    \end{tikzpicture}
}%
}
    \end{center}
    in the category of simple graphs \SimpleGraph, and let matches $m$ restrict to monic matches.
    This rule folds an edge into a loop, similar to \cite[Example~5.2]{overbeek2024termination}, but extended with the addition of two fresh nodes in the right-hand side.
    We prove termination using the weighted type graph $\awtg = \wtg$ given by
    \begin{center}
        \scalebox{1}{{%
\newcommand{\nodex}{\vertex{x}{cblue!20}}%
\newcommand{\nodey}{\vertex{y}{cgreen!20}}%
\newcommand{\nodez}{\vertex{z}{corange!20}}%
\newcommand{\nodew}{\vertex{w}{cred!20}}%
\newcommand{\nodectx}{\vertex{c}{white}}%
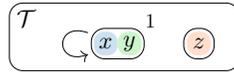
\begin{tikzpicture}[baseline=-10mm,->,node distance=12mm,n/.style={},epattern/.style={very thick},mred/.style={cred,very thick},mblue/.style={cblue,very thick},mgreen/.style={cdarkgreen,very thick},mpurple/.style={cpurple,very thick,densely dotted}]
        \graphbox{$\awtg$}{84mm}{0mm}{34mm}{10mm}{5mm}{4mm}{
          \node [npattern] (x) at (-6mm,-10mm) {\nodex\nodey};
          \annotate{x}{$1$}
          \node [npattern] (y) at (6mm,-10mm) {\nodez};
          \draw [->] (x) to[thinloop=180,distance=1.5em] node [left,label] {} (x);
        }
\end{tikzpicture}
}%
}
    \end{center}
    over the tropical semiring.
    So $\elements$ consists of a single morphism $e$ mapping $\nodeGraph{}$ onto $x$ in $T$, and $\ww(e) = 1$.
    (See Remark~\ref{rem:traceable:simplegraph} on strong traceability in \SimpleGraph.)
    As the context closure $\flower_\arule$ for $\arule$ and $T$ we use the morphism $\flower_\arule: L \to T$ as indicated by the node names ($\{x\}$ and $\{y\}$ are mapped onto $\{x,y\}$).
    
    Then 
      $\weighti[\mathcal{T}]{t_K}{l} = 2 \ostricti 1 = 3 \oplus 2 \oplus 2 \oplus 1 = \weighti[\mathcal{T}]{t_K}{r}$
    for $t_K = \flower_\arule \circ l$.
    For every other morphism $t_K: K \to T$ it holds that 
    $\set{{-} \circ l = t_K} = \emptyset = \set{{-} \circ r = t_K}$.
    Thus~$\arule$ is uniformly decreasing,
    and Theorem~\ref{thm:termination} yields termination of $\arule$.
\end{example}

\begin{example}[Reconfiguration on simple graphs]
    Example~\ref{ex:monic:triangle} can also be considered on \SimpleGraph with monic matching. Assume that we add in $K$ an edge from $\{x\}$ to $\{y\}$ (this does not change the rewrite relation). Then $l$ is regular monic, and we have strong traceability for the edges along the left pushout squares of rewrite steps.
    Hence all left squares are weighable, and all right squares are bounded above by the type graph in Example~\ref{ex:monic:triangle}; so the same type graph proves also termination on \SimpleGraph.
\end{example}

We now consider some examples in \Graph using unrestricted matching.

\begin{example}[Unrestricted matching]
\label{example:plump:string}
    Consider the following double pushout rewrite system  
    \begin{center}
        $\rho = $ \scalebox{1}{{%
\newcommand{\nodex}{\vertex{x}{cblue!20}}%
\newcommand{\nodey}{\vertex{y}{cgreen!20}}%
\newcommand{\nodez}{\vertex{z}{corange!20}}%
\newcommand{\nodew}{\vertex{w}{cred!20}}%
\newcommand{\nodectx}{\vertex{c}{white}}%
      \begin{tikzpicture}[->,node distance=12mm,n/.style={},epattern/.style={very thick},mred/.style={cred,very thick},mblue/.style={cblue,very thick},mgreen/.style={cdarkgreen,very thick},mpurple/.style={cpurple,very thick,densely dotted},baseline=-5.5mm]
        \graphbox{$L$}{0mm}{0mm}{35mm}{9mm}{0.5mm}{-5mm}{
          \node [npattern] (x) at (-10mm,0mm) {\nodex};
          \node [npattern] (y) at (0mm,0mm) {\nodey};
          \node [npattern] (z) at (10mm,0mm) {\nodez};
          \draw (x) to node [above] {$a$} (y);
          \draw (y) to node [above] {$b$} (z);
        }
        \graphbox{$K$}{44mm}{0mm}{35mm}{9mm}{0.5mm}{-5mm}{
          \node [npattern] (x) at (-10mm,0mm) {\nodex};
          \node [npattern] (z) at (10mm,0mm) {\nodez};
        }
        \graphbox{$R$}{88mm}{0mm}{35mm}{9mm}{0.5mm}{-5mm}{
          \node [npattern] (x) at (-10mm,0mm) {\nodex};
          \node [npattern] (y) at (0mm,0mm) {\nodey};
          \node [npattern] (z) at (10mm,0mm) {\nodez};
          \draw (x) to node [above] {$a$} (y);
          \draw (y) to node [above] {$c$} (z);
        }
        \begin{scope}[nodes={fill=white,label,inner sep=0.5mm}]
        \draw [->] (42.5mm,-5mm) to node {$l$} ++(-6mm,0mm);
        \draw [->] (80.5mm,-5mm) to node {$r$} ++(6mm,0mm);
        \end{scope}
      \end{tikzpicture}
}%
}
    \end{center}
    \begin{center}
        $\tau = $ \scalebox{1}{{%
\newcommand{\nodex}{\vertex{x}{cblue!20}}%
\newcommand{\nodey}{\vertex{y}{cgreen!20}}%
\newcommand{\nodez}{\vertex{z}{corange!20}}%
\newcommand{\nodew}{\vertex{w}{cred!20}}%
\newcommand{\nodectx}{\vertex{c}{white}}%
      \begin{tikzpicture}[->,node distance=12mm,n/.style={},epattern/.style={very thick},mred/.style={cred,very thick},mblue/.style={cblue,very thick},mgreen/.style={cdarkgreen,very thick},mpurple/.style={cpurple,very thick,densely dotted},baseline=-5.5mm]
        \graphbox{$L$}{0mm}{0mm}{35mm}{9mm}{0.5mm}{-5mm}{
          \node [npattern] (x) at (-10mm,0mm) {\nodex};
          \node [npattern] (y) at (0mm,0mm) {\nodey};
          \node [npattern] (z) at (10mm,0mm) {\nodez};
          \draw (x) to node [above] {$c$} (y);
          \draw (y) to node [above] {$d$} (z);
        }
        \graphbox{$K$}{44mm}{0mm}{35mm}{9mm}{0.5mm}{-5mm}{
          \node [npattern] (x) at (-10mm,0mm) {\nodex};
          \node [npattern] (z) at (10mm,0mm) {\nodez};
        }
        \graphbox{$R$}{88mm}{0mm}{35mm}{9mm}{0.5mm}{-5mm}{
          \node [npattern] (x) at (-10mm,0mm) {\nodex};
          \node [npattern] (y) at (0mm,0mm) {\nodey};
          \node [npattern] (z) at (10mm,0mm) {\nodez};
          \draw (x) to node [above] {$d$} (y);
          \draw (y) to node [above] {$b$} (z);
        }
        \begin{scope}[nodes={fill=white,label,inner sep=0.5mm}]
        \draw [->] (42.5mm,-5mm) to node {$l$} ++(-6mm,0mm);
        \draw [->] (80.5mm,-5mm) to node {$r$} ++(6mm,0mm);
        \end{scope}
      \end{tikzpicture}
}%
}
    \end{center}
    from~\cite[Example 3]{plump2018modular} and \cite[Example 3.8]{plump1995on}
    in \Graph with unrestricted matching.
    These are graph transformation versions of the string rewrite rules $ab \to ac$ and $cd \to db$, respectively.
    We prove termination in two steps using the type graphs 
    \begin{center}
        \scalebox{1}{{%
\newcommand{\nodex}{\vertex{x}{cblue!20}}%
\newcommand{\nodey}{\vertex{y}{cgreen!20}}%
\newcommand{\nodez}{\vertex{z}{corange!20}}%
\newcommand{\nodew}{\vertex{w}{cpurple!20}}%
\newcommand{\nodeu}{\vertex{u}{white}}%
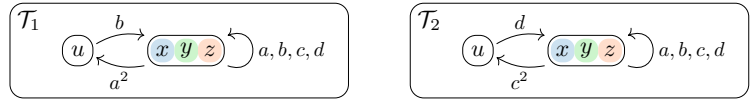
\begin{tikzpicture}[baseline=-10mm,->,node distance=12mm,n/.style={},epattern/.style={very thick},mred/.style={cred,very thick},mblue/.style={cblue,very thick},mgreen/.style={cdarkgreen,very thick},mpurple/.style={cpurple,very thick,densely dotted}]
        \graphbox{$\awtg_1$}{84mm}{0mm}{50mm}{14mm}{1mm}{3mm}{
          \node [npattern] (u) at (-16mm,-10mm) {\nodeu};
          \node [npattern] (xz) at (0mm,-10mm) {\nodex\nodey\nodez};
          \draw [->] (u) to[bend left=20] node [above,label] {$b$} (xz);
          \draw [->] (xz) to[bend left=20] node [below,label] {$a^2$} (u);
          \draw [->] (xz) to[thinloop=0,distance=1.5em] node [right,label] {$a,b,c,d$} (xz);
        }
\end{tikzpicture}
}%
}\quad\quad
        \scalebox{1}{{%
\newcommand{\nodex}{\vertex{x}{cblue!20}}%
\newcommand{\nodey}{\vertex{y}{cgreen!20}}%
\newcommand{\nodez}{\vertex{z}{corange!20}}%
\newcommand{\nodew}{\vertex{w}{cpurple!20}}%
\newcommand{\nodeu}{\vertex{u}{white}}%
\begin{tikzpicture}[baseline=-10mm,->,node distance=12mm,n/.style={},epattern/.style={very thick},mred/.style={cred,very thick},mblue/.style={cblue,very thick},mgreen/.style={cdarkgreen,very thick},mpurple/.style={cpurple,very thick,densely dotted}]
        \graphbox{$\awtg_2$}{84mm}{0mm}{50mm}{14mm}{1mm}{3mm}{
          \node [npattern] (u) at (-16mm,-10mm) {\nodeu};
          \node [npattern] (xz) at (0mm,-10mm) {\nodex\nodey\nodez};
          \draw [->] (u) to[bend left=20] node [above,label] {$d$} (xz);
          \draw [->] (xz) to[bend left=20] node [below,label] {$c^2$} (u);
          \draw [->] (xz) to[thinloop=0,distance=1.5em] node [right,label] {$a,b,c,d$} (xz);
        }
\end{tikzpicture}
}%
}
    \end{center}
    over the arithmetic semiring.
    Here the weights of the weighted elements are indicated as superscripts (both type graphs have one weighted element with domain $\edgeGraph$ and weight $2$).
    
    First, we take the weighted type graph $\awtg_1$.
    Consider the rule $\rho$ together with the context closure $\flower_\rho: L \to T_1$ mapping all nodes onto $\{x,y,z\}$ in $T_1$.
    For $t_K = \flower_\rho \circ l$ we have 
    \begin{align*}
      \weighti[\mathcal{T}_1]{t_K}{l} = 1 + 2 \ostricti 1 = \weighti[\mathcal{T}_1]{t_K}{r}. 
    \end{align*}
    Here weight $1$ arises from those morphisms mapping all nodes onto $\{x,y,z\}$ in $T_1$;
    then there are no weighted elements in the image of the morphism and the empty product is $1$.
    For all other $t_K : K \to T$ we have $\weighti[\mathcal{T}_1]{t_K}{l} = 0 \oweaki 0 = \weighti[\mathcal{T}_1]{t_K}{r}$. 
    Thus $\arule$ is closure decreasing.
    
    For the rule $\tau$, $L$ and $R$ admit only one morphism with codomain $T_1$. So it is easy to see that $\weighti[\mathcal{T}_1]{t_K}{l} \oweaki \weighti[\mathcal{T}_1]{t_K}{r}$ for all $t_K : K \to T$, and thus $\tau$ is weakly decreasing. 
    
    As a consequence, by Theorem~\ref{thm:termination}, $\arule$ is terminating relative to $\tau$. 
    Hence it suffices to prove termination of $\tau$ which follows analogously using $\awtg_2$.  
\end{example}

\begin{example}[Tree counter]
\label{example:tree:counter}
\begin{samepage}
    Consider the system 
    {%
\newcommand{\crule}[1]{\scalebox{0.97}{
\begin{tikzpicture}[baseline=-5mm,->,node distance=12mm,n/.style={},epattern/.style={very thick},mred/.style={cred,very thick},mblue/.style={cblue,very thick},mgreen/.style={cdarkgreen,very thick},mpurple/.style={cpurple,very thick,densely dotted}]
#1
\end{tikzpicture}%
}}
\newcommand{\nodex}{\vertex{x}{cblue!20}}%
\newcommand{\nodey}{\vertex{y}{cgreen!20}}%
\newcommand{\nodez}{\vertex{z}{corange!20}}%
\newcommand{\nodew}{\vertex{w}{cred!20}}%
\newcommand{\nodem}{\vertex{}{white}}%
\newcommand{\tree}[3]{
    \node [npattern] (x) at (-7.5mm,4mm) {\nodex};
    \node [npattern] (y) at (-7.5mm,-4mm) {\nodey};
    \node [npattern] (z) at (7.5mm,0mm) {\nodez};
    \node [npattern,circle,minimum size=3mm] (m) at (0mm,0mm) {};
    \draw (z) to node[above] {$#1$} (m);
    \draw (m) to node[above,pos=0.3] {$#2$} (x);
    \draw (m) to node[below,pos=0.3] {$#3$} (y);
    }%
\newcommand{\treerule}[6]{
    \graphbox{}{0mm}{0mm}{23mm}{16mm}{0mm}{-8mm}{
      \tree{#1}{#2}{#3}
    }
    \graphbox{}{26mm}{0mm}{15mm}{16mm}{0mm}{-8mm}{
      \node [npattern] (x) at (-3.5mm,4mm) {\nodex};
      \node [npattern] (y) at (-3.5mm,-4mm) {\nodey};
      \node [npattern] (z) at (3.5mm,0mm) {\nodez};
    }
    \graphbox{}{44mm}{0mm}{23mm}{16mm}{0mm}{-8mm}{
      \tree{#4}{#5}{#6}
    }
    \begin{scope}[nodes={fill=white,label,inner sep=0.5mm}]
    \draw [->] (25.5mm,-8mm) to ++(-2mm,0mm);
    \draw [->] (41.5mm,-8mm) to ++(2mm,0mm);
    \end{scope}
}
\begin{center}
\crule{
    \node at (-4mm,-4mm) {$\rho_1\,{=}$};
        \graphbox{}{0mm}{0mm}{23mm}{8mm}{0mm}{-4.5mm}{
          \node [npattern] (x) at (-6mm,0mm) {\nodex};
          \node [npattern] (y) at (6mm,0mm) {\nodey};
          \draw (y) to node [above] {$0$} (x);
        }
        \graphbox{}{26mm}{0mm}{15mm}{8mm}{0mm}{-4.5mm}{
          \node [npattern] (x) at (0mm,0mm) {\nodex};
        }
        \graphbox{}{44mm}{0mm}{23mm}{8mm}{0mm}{-4.5mm}{
          \node [npattern] (x) at (-6mm,0mm) {\nodex};
          \node [npattern] (y) at (6mm,0mm) {\nodey};
          \draw (y) to node [above] {$1$} (x);
        }
        \begin{scope}[nodes={fill=white,label,inner sep=0.5mm}]
        \draw [->] (25.5mm,-4mm) to ++(-2mm,0mm);
        \draw [->] (41.5mm,-4mm) to ++(2mm,0mm);
        \end{scope}
}%
\crule{
        \node at (-4mm,-4mm) {$\rho_2\,{=}$};
        \graphbox{}{0mm}{0mm}{23mm}{8mm}{0mm}{-4.5mm}{
          \node [npattern] (x) at (-6mm,0mm) {\nodex};
          \node [npattern] (y) at (6mm,0mm) {\nodey};
          \draw (y) to node [above] {$1$} (x);
        }
        \graphbox{}{26mm}{0mm}{15mm}{8mm}{0mm}{-4.5mm}{
          \node [npattern] (x) at (0mm,0mm) {\nodex};
        }
        \graphbox{}{44mm}{0mm}{23mm}{8mm}{0mm}{-4.5mm}{
          \node [npattern] (x) at (-6mm,0mm) {\nodex};
          \node [npattern] (y) at (6mm,0mm) {\nodey};
          \draw (y) to node [above] {$c$} (x);
        }
        \begin{scope}[nodes={fill=white,label,inner sep=0.5mm}]
        \draw [->] (25.5mm,-4mm) to ++(-2mm,0mm);
        \draw [->] (41.5mm,-4mm) to ++(2mm,0mm);
        \end{scope}
}
\crule{
        \node at (-4mm,-8mm) {$\rho_3\,{=}$};
        \treerule{c}{0}{0}{0}{1}{0}
}%
\crule{
        \node at (-4mm,-8mm) {$\rho_4\,{=}$};
        \treerule{c}{1}{0}{0}{c}{0}
}
\crule{
        \node at (-4mm,-8mm) {$\rho_5\,{=}$};
        \treerule{c}{0}{1}{0}{1}{1}
}%
\crule{
        \node at (-4mm,-8mm) {$\rho_6\,{=}$};
        \treerule{c}{1}{1}{0}{c}{1}
}
\end{center}
}%

    \noindent
    from~\cite[Example 6]{bruggink2015} in \Graph with unrestricted matching.
    This system implements a counter on trees.
\end{samepage}%
    We prove termination in two steps using the type graphs 
    \begin{center}
        \scalebox{1}{{%
\newcommand{\nodex}{\vertex{x}{cblue!20}}%
\newcommand{\nodey}{\vertex{y}{cgreen!20}}%
\newcommand{\nodez}{\vertex{z}{corange!20}}%
\newcommand{\nodew}{\vertex{w}{white}}%
\newcommand{\nodeu}{\vertex{u}{white}}%
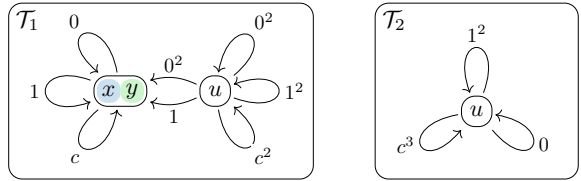
\begin{tikzpicture}[baseline=-10mm,->,node distance=12mm,n/.style={},epattern/.style={very thick},mred/.style={cred,very thick},mblue/.style={cblue,very thick},mgreen/.style={cdarkgreen,very thick},mpurple/.style={cpurple,very thick,densely dotted}]
        \graphbox{$\awtg_1$}{84mm}{0mm}{45mm}{26mm}{0mm}{-3mm}{
          \node [npattern] (0) at (-6mm,-10mm) {\nodex\nodey};
          \node [npattern] (1) at (8mm,-10mm) {\nodeu};
          
          \draw [->] (0) to[thinloop=120] node [above left,label] {$0$} (0);
          \draw [->] (1) to[bend left=-20] node [above,label] {$0^2$} (0);
          \draw [->] (1) to[thinloop=60] node [above right,label] {$0^2$} (1);
          
          \draw [->] (0) to[thinloop=180] node [left,label] {$1$} (0);          
          \draw [->] (1) to[bend left=20] node [below,label] {$1$} (0);
          \draw [->] (1) to[thinloop=0] node [right,label] {$1^2$} (1);

          \draw [->] (0) to[thinloop=240] node [below left,label] {$c$} (0);
          \draw [->] (1) to[thinloop=-60] node [below right,label,yshift=1mm] {$c^2$} (1);
        }
\end{tikzpicture}%
}%
}\quad\quad
        \scalebox{1}{{%
\newcommand{\nodex}{\vertex{x}{cblue!20}}%
\newcommand{\nodey}{\vertex{y}{cgreen!20}}%
\newcommand{\nodez}{\vertex{z}{corange!20}}%
\newcommand{\nodew}{\vertex{w}{white}}%
\newcommand{\nodeu}{\vertex{u}{white}}%
\begin{tikzpicture}[baseline=-10mm,->,node distance=12mm,n/.style={},epattern/.style={very thick},mred/.style={cred,very thick},mblue/.style={cblue,very thick},mgreen/.style={cdarkgreen,very thick},mpurple/.style={cpurple,very thick,densely dotted}]
        \graphbox{$\awtg_2$}{84mm}{0mm}{30mm}{26mm}{0mm}{-6mm}{
          \node [npattern] (0) at (0mm,-10mm) {\nodeu};
          
          \draw [->] (0) to[thinloop=-30] node [right,label] {$0$} (0);
          
          \draw [->] (0) to[thinloop=90] node [above,label] {$1^2$} (0);

          \draw [->] (0) to[thinloop=210] node [left,label] {$c^3$} (0);
        }
\end{tikzpicture}%
}%
}
    \end{center}
    over the arithmetic semiring. The superscripts indicate the weights of the weighted elements (all weighted elements have domain $\edgeGraph$).
    
    First, we take the weighted type graph $\awtg_1$.
    Consider the rule $\rho_1$ together with the context closure $\flower_\rho: L \to T_1$ mapping all nodes onto $\{x,y\}$ in $T_1$.
    For $t_K = \flower_\rho \circ l$ we have 
    \begin{align*}
      \weighti[\mathcal{T}_1]{t_K}{l} = 1 + 2 \ostricti 1 + 1 = \weighti[\mathcal{T}_1]{t_K}{r}. 
    \end{align*}
    For $t_K$ mapping $\{x\}$ to $\{u\}$, we get 
    $\weighti[\mathcal{T}_1]{t_K}{l} = 2 \oweaki 2 = \weighti[\mathcal{T}_1]{t_K}{r}$.
    Thus $\rho_1$ is closure decreasing.
    Likewise, $\rho_2$ is also closure decreasing and all other rules are weakly decreasing.
    By Theorem~\ref{thm:termination} it suffices to prove termination without the rules $\rho_1$ and $\rho_2$.

    Termination of the remaining rules is established by the weighted type graph $\awtg_2$.
    The rules $\rho_3$, $\rho_4$, $\rho_5$ and $\rho_6$ are all closure decreasing with respect to $\awtg_2$.
    As a consequence, the system is terminating by Theorem~\ref{thm:termination}. 
\end{example}

The following example shows that the technique also can be fruitfully applied for rules $\arule: L \leftarrow K \to R$ for which $L$ is a subgraph of $R$, that is, $L \mono R$.

\begin{example}[Morphism counting]
\label{ex:empty}
    Let $\arule$ be the double pushout rule
    \begin{center}
        \scalebox{1}{{%
\newcommand{\nodex}{\vertex{x}{cblue!20}}%
\newcommand{\nodey}{\vertex{y}{cgreen!20}}%
\newcommand{\nodez}{\vertex{z}{corange!20}}%
\newcommand{\nodew}{\vertex{w}{cred!20}}%
\newcommand{\nodectx}{\vertex{c}{white}}%
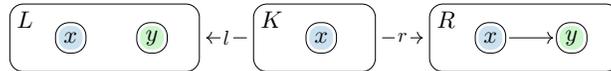
\begin{tikzpicture}[baseline=-5mm,->,node distance=12mm,n/.style={},epattern/.style={very thick},mred/.style={cred,very thick},mblue/.style={cblue,very thick},mgreen/.style={cdarkgreen,very thick},mpurple/.style={cpurple,very thick,densely dotted}]
        \graphbox{$L$}{0mm}{0mm}{28mm}{10mm}{1mm}{-5mm}{
          \node [npattern] (x) at (-6mm,0mm) {\nodex};
          \node [npattern] (y) at (6mm,0mm) {\nodey};
        }
        \graphbox{$K$}{36mm}{0mm}{18mm}{10mm}{1mm}{-5mm}{
          \node [npattern] (x) at (0mm,0mm) {\nodex};
        }
        \begin{scope}[opacity=1]        
        \graphbox{$R$}{62mm}{0mm}{28mm}{10mm}{1mm}{-5mm}{
          \node [npattern] (x) at (-6mm,0mm) {\nodex};
          \node [npattern] (y) at (6mm,0mm) {\nodey};
          \draw (x) to (y);
        }
        \end{scope}
        \begin{scope}[nodes={fill=white,label,inner sep=0.5mm}]
        \draw [->] (35mm,-5mm) to node {$l$} ++(-6mm,0mm);
        \draw [->] (55mm,-5mm) to node {$r$} ++(6mm,0mm);
        \end{scope}
    \end{tikzpicture}
}%
}
    \end{center}
    in the category of graphs \Graph with unrestricted matching.
    Intuitively, this rule connects an isolated node with an edge.
    We prove termination using the weighted type graph $\awtg = \wtg$ given by
    \begin{center}
        \scalebox{1}{{%
\newcommand{\nodex}{\vertex{x}{cblue!20}}%
\newcommand{\nodey}{\vertex{y}{cgreen!20}}%
\newcommand{\nodez}{\vertex{z}{corange!20}}%
\newcommand{\nodew}{\vertex{w}{cred!20}}%
\newcommand{\nodectx}{\vertex{c}{white}}%
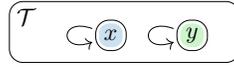
\begin{tikzpicture}[baseline=-10mm,->,node distance=12mm,n/.style={},epattern/.style={very thick},mred/.style={cred,very thick},mblue/.style={cblue,very thick},mgreen/.style={cdarkgreen,very thick},mpurple/.style={cpurple,very thick,densely dotted}]
        \graphbox{$\awtg$}{84mm}{0mm}{34mm}{9mm}{4mm}{5mm}{
          \node [npattern] (x) at (-6mm,-10mm) {\nodex};
          \node [npattern] (y) at (6mm,-10mm) {\nodey};
          \draw [->] (x) to[thinloop=180,distance=1.5em] node [left,label] {} (x);
          \draw [->] (y) to[thinloop=180,distance=1.5em] node [left,label] {} (y);
        }
\end{tikzpicture}
}%
}
    \end{center}
    over the arithmetic semiring.
    Here $\elements = \emptyset$. This means that every morphism into $T$ has weight $1$ (the empty product) and thus the weight of an object $X \in \ob{\categoryc}$ is $\#\homset{X}{T}$, the number of morphisms from $X$ to $T$.
    As the context closure $\flower_\arule$ for $\arule$ and $T$ we use the inclusion $\flower_\arule: L \mono T$ as indicated by the node names.
    
    For every morphism $t_K : K \to T$ we have 
    \begin{align*}
      \weighti[\mathcal{T}]{t_K}{l} = 1 + 1 = 2 \ostricti 1 = \weighti[\mathcal{T}]{t_K}{r}  
    \end{align*}
    Thus~$\arule$ is closure decreasing,
    and Theorem~\ref{thm:termination} yields termination of $\arule$.
\end{example}

Finally, we consider a hypergraph example with unrestricted matching. Our method can prove relative termination of the first rule, but termination of the second rule does not seem to be provable with our method.

\begin{example}[Limitations]
\label{ex:limitations}
    Consider the following hypergraph rewriting system
    \begin{center}
        $\rho = $\, \scalebox{1}{{%
\newcommand{\context}{
}%
\newcommand{\nodex}{\vertex{x}{cblue!20}}%
\newcommand{\nodey}{\vertex{y}{cgreen!20}}%
\newcommand{\nodez}{\vertex{z}{corange!20}}%
\newcommand{\nodew}{\vertex{w}{cred!20}}%
\newcommand{\nodectx}{\vertex{c}{white}}%
      \begin{tikzpicture}[->,node distance=12mm,n/.style={},epattern/.style={very thick},mred/.style={cred,very thick},mblue/.style={cblue,very thick},mgreen/.style={cdarkgreen,very thick},mpurple/.style={cpurple,very thick,densely dotted},baseline=-28mm]
        \graphbox{$L$}{0mm}{-11mm}{35mm}{33mm}{2mm}{-5mm}{
          \node [npattern] (x) at (0mm,0mm) {\nodex};
          \node [npattern] (y) at (-8mm,-15mm) {\nodey};
          \node [npattern] (z) at (8mm,-15mm) {\nodez};
          \node [rectangle,draw=black,fill=white] (sx) [below of=x,node distance=7.5mm] {$+$};
          \draw [-] (x) to  (sx);
          \draw [->] (sx.south west) to  (y);
          \draw [->] (sx.south east) to  (z);
          \node [rectangle,draw=black,fill=white] (0) [below of=z,node distance=8mm] {$0$}; 
          \draw [-] (z) to  (0);
          \context
        }
        \graphbox{$K$}{44mm}{-11mm}{35mm}{33mm}{2mm}{-5mm}{
          \node [npattern] (x) at (0mm,0mm) {\nodex};
          \node [npattern] (y) at (-8mm,-15mm) {\nodey};
          \node [npattern] (z) at (8mm,-15mm) {\nodez};
          \node [rectangle,draw=black,fill=white] (0) [below of=z,node distance=8mm] {$0$}; 
          \draw [-] (z) to  (0);
          \context
        }
        \begin{scope}[opacity=1]        
        \graphbox{$R$}{88mm}{-11mm}{35mm}{33mm}{2mm}{-5mm}{
          \node [npattern] (x) at (0mm,0mm) {\nodex\nodey};
          \node [npattern] (z) at (0mm,-15mm) {\nodez};
          \node [rectangle,draw=black,fill=white] (0) [below of=z,node distance=8mm] {$0$}; 
          \draw [-] (z) to  (0);
          \context
        }
        \end{scope}
        \begin{scope}[nodes={fill=white,label,inner sep=0.5mm}]
        \draw [->] (42.5mm,-27mm) to node {$l$} ++(-6mm,0mm);
        \draw [->] (80.5mm,-27mm) to node {$r$} ++(6mm,0mm);
        \end{scope}
      \end{tikzpicture}
}%
}
    \end{center}
    \begin{center}
        $\tau = $\, \scalebox{1}{{%
\newcommand{\context}{
}%
\newcommand{\nodex}{\vertex{x}{cblue!20}}%
\newcommand{\nodey}{\vertex{y}{cgreen!20}}%
\newcommand{\nodez}{\vertex{z}{corange!20}}%
\newcommand{\nodew}{\vertex{w}{cred!20}}%
\newcommand{\nodectx}{\vertex{c}{white}}%
      \begin{tikzpicture}[->,node distance=12mm,n/.style={},epattern/.style={very thick},mred/.style={cred,very thick},mblue/.style={cblue,very thick},mgreen/.style={cdarkgreen,very thick},mpurple/.style={cpurple,very thick,densely dotted},baseline=-28mm]
        \graphbox{$L$}{0mm}{-11mm}{35mm}{33mm}{2mm}{-5mm}{
          \node [npattern] (x) at (-8mm,0mm) {\nodex};
          \node [npattern] (y) at (8mm,0mm) {\nodey};
          \node [npattern] (z) at (0mm,-15mm) {\nodez};
          \node [rectangle,draw=black,fill=white] (sx) [below of=x,node distance=7.5mm] {$s$};
          \draw [-] (x) to  (sx);
          \draw [->] (sx.south east) to  (z);
          \node [rectangle,draw=black,fill=white] (sy) [below of=y,node distance=7.5mm] {$s$};
          \draw [-] (y) to  (sy);
          \draw [->] (sy.south west) to  (z);
          \node [rectangle,draw=black,fill=white] (0) [below of=z,node distance=8mm] {$0$};
          \draw [-] (z) to  (0);
          \context
        }
        \graphbox{$K$}{44mm}{-11mm}{35mm}{33mm}{2mm}{-5mm}{
          \node [npattern] (x) at (-8mm,0mm) {\nodex};
          \node [npattern] (y) at (8mm,0mm) {\nodey};
          \node [npattern] (z) at (0mm,-15mm) {\nodez};
          \node [rectangle,draw=black,fill=white] (sx) [below of=x,node distance=7.5mm] {$s$};
          \draw [-] (x) to  (sx);
          \draw [->] (sx.south east) to  (z);
          \node [rectangle,draw=black,fill=white] (0) [below of=z,node distance=8mm] {$0$};
          \draw [-] (z) to  (0);
          \context
        }
        \begin{scope}[opacity=1]        
        \graphbox{$R$}{88mm}{-11mm}{35mm}{33mm}{2mm}{-5mm}{
          \node [npattern] (x) at (-7mm,0mm) {\nodex};
          \node [npattern] (y) at (7mm,0mm) {\nodey};
          \node [npattern] (z) at (-7mm,-15mm) {\nodez};
          \node [npattern,minimum size=5mm] (n) at (7mm,-15mm) {};
          \node [rectangle,draw=black,fill=white] (sx) [below of=x,node distance=7.5mm] {$s$};
          \draw [-] (x) to  (sx);
          \draw [->] (sx) to  (z);
          \node [rectangle,draw=black,fill=white] (0) [below of=z,node distance=8mm] {$0$};
          \draw [-] (z) to  (0);
          \node [rectangle,draw=black,fill=white] (sy) [below of=y,node distance=7.5mm] {$s$};
          \draw [-] (y) to  (sy);
          \draw [->] (sy) to  (n);
          \node [rectangle,draw=black,fill=white] (0b) [below of=n,node distance=8mm] {$0$};
          \draw [-] (n) to  (0b);
          \context
        }
        \end{scope}
        \begin{scope}[nodes={fill=white,label,inner sep=0.5mm}]
        \draw [->] (42.5mm,-27mm) to node {$l$} ++(-6mm,0mm);
        \draw [->] (80.5mm,-27mm) to node {$r$} ++(6mm,0mm);
        \end{scope}
      \end{tikzpicture}
}%
}
    \end{center}
    from~\cite[Example 6]{plump2018modular} with unrestricted matching.

    Termination of $\rho$ relative to $\tau$ follows  using the following weighted type graph 
    \begin{center}
        \scalebox{1}{{%
\newcommand{\context}{
}%
\newcommand{\nodex}{\vertex{x}{cblue!20}}%
\newcommand{\nodey}{\vertex{y}{cgreen!20}}%
\newcommand{\nodez}{\vertex{z}{corange!20}}%
\newcommand{\nodew}{\vertex{w}{cblue!20}}%
\newcommand{\nodectx}{\vertex{c}{white}}%
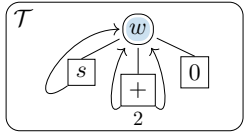
      \begin{tikzpicture}[->,node distance=12mm,n/.style={},epattern/.style={very thick},mred/.style={cred,very thick},mblue/.style={cblue,very thick},mgreen/.style={cdarkgreen,very thick},mpurple/.style={cpurple,very thick,densely dotted},baseline=-28mm]
        \graphbox{$\awtg$}{0mm}{-11mm}{35mm}{20mm}{2mm}{-4mm}{
          \node [npattern] (x) at (0mm,0mm) {\nodew};
          \node [rectangle,draw=black,fill=white] (sx) [below of=x,node distance=9mm] {$+$}; 
          \node [label,anchor=north] at (sx.south) {$2$};
          \draw [-] (x) to  (sx);
          \draw [->] (sx.south west) to[out=-135,in=-125] (x);
          \draw [->] (sx.south east) to[out=-45,in=-55] (x);
          \node [rectangle,draw=black,fill=white] (0) [below right of=x,xshift=2mm,node distance=9mm] {$0$}; 
          \draw [-] (x) to  (0.north);
          \node [rectangle,draw=black,fill=white] (s) [below left of=x,xshift=-2mm,node distance=9mm] {$s$}; 
          \draw [-] (x) to  (s.north);
          \draw [->] (s.south west) to[out=-145,in=-180,looseness=2]  (x);
        }
      \end{tikzpicture}
}%
}
    \end{center}
    over the arithmetic semiring by Theorem~\ref{thm:termination}.
    Here the `$+$' hyperedge has weight $2$.
    
    However, $\tau$ is an interesting rule for which we could not find a termination proof using weighted type graphs over the arctic, tropical or arithmetic semiring.
    First, note that there exists an epimorphism $e: R \epi L$ with $l = e\circ r$. 
    In such situations, no weighted type graph over the arithmetic semiring can make the rule strictly decreasing.
    Second, it seems that the arctic and tropical semiring are also not applicable here. Intuitively, one might try to use the tropical semiring since the additional freedom in mapping $R$ could lead to a lower weight. 
    However, for this rule, there are morphisms $h_1, h_2: L \to R$ where $h_1$ maps $L$ onto the left component of $R$, and $h_2$ maps $L$ onto the right component of $R$. 
    So, for any mapping $t: R \to T$, either $t \circ h_1$ or $t \circ h_2$ will be at least as `light'.
    
    In~\cite{plump2018modular}, it is argued that this rule is terminating since the sum of the squares of in-degrees is decreasing in each step. This argument can be automated, for instance, using the technique proposed in~\cite[Example~5.9]{overbeek2024termination}: see Section~\ref{discussion:termination:of:graph:transformation} for a related discussion.
\end{example}

\begin{samepage}
Example~\ref{ex:limitations} illustrates a limitation of our technique:
\begin{open}\label{open:epi}
    We conjecture that our technique cannot be applied for rules $\arule: L \stackrel{l}{\leftarrow} K \stackrel{r}{\to} R$ for which there exists an epimorphism $e: R \epi L$ with $l = e\circ r$.
    Can the technique be enhanced to deal with such rules?
\end{open}
\end{samepage}

\section{Implementation and Benchmarks}
\label{sec:benchmarks}

We have implemented our techniques in \textsf{graphTT-wtg} (short for \textbf{Graph} \textbf{T}ermination \textbf{T}ool -- \textbf{w}eighted \textbf{t}ype \textbf{g}raphs) for automatically proving 
termination of graph transformation. 
The tool is written in Scala and employs Z3~\cite{z3} to solve the constraint systems.\footnote{\textsf{graphTT-wtg} is a fork of an early version of \textsf{graphTT}~\cite{overbeek2024graphtt} described in \cite[Section 7]{overbeek2024termination}. The two have since diverged.} %
The source code (including examples) are available for download at
\begin{center}
    {\small\url{http://joerg.endrullis.de/downloads/graphTT_wtg.tar.gz}}
\end{center}

\newcommand{\iop}{I^\text{op}}
The tool allows to specify the underlying category using \emph{copresheaves}, i.e., functors $F: \catname{I} \to \Set$\footnote{We restrict to acyclic copresheaves, i.e., copresheaves that are acyclic if one ignores the identity morphisms.}, which are equivalent to \emph{unary algebras}~\cite{lowe1993algebraic} and $\mathcal{C}$-sets~\cite{julia}. 
The grammar for specifying the index category $\catname{I}$ is
\begin{align*}
    \textit{IndexCategory} = \; &\textit{Object} \;\ldots\; \textit{Object} \\
    \textit{Object} = \; &\textit{Name}[\boldsymbol{[}\underbrace{\textit{Label},\ldots,\textit{Label}}_{\text{labels}}\boldsymbol{]}](\underbrace{\textit{Name},\ldots,\textit{Name}}_{\text{arguments}})[!]
\end{align*}
where \textit{Name} is the name of an object in $\catname{I}$, and the `arguments' are the targets of outgoing arrows. Optionally, an object can be equipped with a set of labels which corresponds to a slice category construction.
\textsf{graphTT-wtg} can thereby be applied to the multigraph examples~\cite{bruggink2014,bruggink2015}, and the hypergraph examples~\cite{plump2018modular}.

Moreover, the tool supports what could be called `partially simple' presheaves in which some of the element sets are simple (the elements are uniquely determined by their `argument values'). To this end, an object declaration can be followed by `!'.
For instance
\begin{align*}
  V &&
  \text{edge}(V,V) &&
  \text{plus}(V,V,V)! &&
  \text{flag}[a,b](V)
\end{align*}
describes presheaves having a set of nodes $V$,
binary edges,
ternary hyperedges `plus' between nodes of $V$, 
and unary hyperedges `flag' labelled with labels from $\{a,b\}$.
The binary edges allow for multiple parallel edges. The ternary `plus' is declared as simple, meaning that there cannot be two `plus' edges between the same tuple of nodes (the order matters).

\begin{table}[t]
\renewcommand{\arraystretch}{1.2}
\centering
\begin{tabular}{|c|c|c|c|c|c|c|c|c|c|c|c|c|c|c|}
\hline
 & \;\;A\;\; & \;\;T\;\; & \;\;N\;\; & \;A,T\; & \;A,N\; & \;T,N\; & \;\textbf{A,T,N}\;\\
\hline
Example~\ref{ex:monic:loop} &  &  & 0.08 &  & 0.06 & 0.06 & \textbf{0.09}\\
Example~\ref{ex:monic:triangle} &  &  & 0.12 &  & 0.11 & 0.11 & \textbf{0.12}\\
\quad{}Example~\ref{ex:monic:simple}\quad{} &  & 0.03 &  & 0.03 &  & 0.03 & \textbf{0.04}\\
Example~\ref{ex:empty} & 0.3 &  & 0.06 & 0.04 & 0.04 & 0.05 & \textbf{0.04}\\
\cite[Example 3.8]{plump1995on} & 0.17 & 0.19 & 0.54 & 0.20 & 0.20 & 0.21 & \textbf{0.21}\\
\cite[Figure 10]{plump2018modular} & 0.43 & 0.42 & 0.49 & 0.46 & 0.51 & 0.48 & \textbf{0.52}\\
\cite[Example 3]{plump2018modular} & 0.11 & 0.09 & 0.26 & 0.1 & 0.12 & 0.11 & \textbf{0.12}\\
\cite[Example 6]{plump2018modular}  &  &  &  &  &  &  &  \\
\cite[Example 4]{bruggink2015} & 0.30 & 0.35 & 0.56 & 0.24 & 0.26 & 0.31 & \textbf{0.27}\\
\cite[Example 5]{bruggink2015} &  &  & 1.11 &  & 0.89 & 0.82 & \textbf{0.83}\\
\cite[Example 6]{bruggink2015} &  &  & 2.44 &  & 1.20 & 1.25 & \textbf{1.34}\\
\cite[Example 1]{bruggink2014} & 0.10 &  & 0.18 & 0.09 & 0.11 & 0.19 & \textbf{0.11}\\
\cite[Example 4]{bruggink2014} & 0.07 & 0.08 & 0.26 & 0.09 & 0.08 & 0.10 & \textbf{0.10}\\
\cite[Example 5]{bruggink2014} & 0.61 & 0.50 & 2.39 & 0.49 & 0.68 & 0.57 & \textbf{0.52}\\
\cite[Example 6]{bruggink2014} &  &  &  & 0.23 &  & 0.47 & \textbf{0.39}\\
\hline
\end{tabular}%
\medskip
\caption{Benchmark for the examples from \cite{bruggink2014,bruggink2015,plump1995on,plump2018modular,overbeek2024termination} and the examples from Section~\ref{sec:examples}. The columns display different semiring configurations: A, T and N stand for arctic, tropical and arithmetic semiring, respectively. Runtimes are indicated in seconds (s). Observe that \cite[Example 6]{bruggink2014} requires either a combination of both arctic with tropical, or tropical with arithmetic semiring, for different iterations of the relative termination proof.}
\label{tabular:benchmarks}
\end{table}

\textsf{graphTT-wtg} supports a simple strategy language of the form
\begin{align*}
    &\hspace*{-7mm}\textit{Strategy} =\\ 
    \; &\textit{Strategy} \mathop{\textbf{;}} \textit{Strategy} &&\text{sequential composition}\\
    \mid \; & \textit{Strategy} \mathop{\boldsymbol{|}} \textit{Strategy} 
    &&\text{parallel composition}\\
    \mid \; & \text{repeat}(\textit{Strategy}) 
    &&\text{repeat while successful}\\
    \mid \; & \text{arctic}(\text{size}=\nat,\text{bits}=\nat,\text{timeout}=\nat)
    &&\text{arctic type graphs}\\ 
    \mid \; & \text{tropical}(\text{size}=\nat,\text{bits}=\nat,\text{timeout}=\nat)
    &&\text{tropical type graphs}\\ 
    \mid \; & \text{arithmetic}(\text{size}=\nat,\text{bits}=\nat,\text{timeout}=\nat)
    &&\text{arithmetic type graphs}
\end{align*}
This enables the user to specify what technique is applied in what order and for how long. 
Each strategy can transform the system by removing rules.
A strategy is considered successful if at least one rule has been removed.
For the basis strategies, that is, arctic, tropical or arithmetic weighted type graphs, `size' specifies the size of the type graph, `bits' the number of bits for encoding the semiring elements, and `timeout' is the the maximum number of seconds to try this technique.

Table~\ref{tabular:benchmarks} shows benchmarks for running \textsf{graphTT-wtg} on the examples from~\cite{bruggink2014,bruggink2015,plump1995on,plump2018modular}\footnote{We have included all uniform termination examples from~\cite{bruggink2014,bruggink2015,plump1995on,plump2018modular}. We did not include \cite[Example 3]{bruggink2014} as this concerns local termination on a particular starting language. We have also not included~\cite[Example 4.1]{plump1995on} since this concerns a category of term graphs (which our tool does not support). The paper~\cite{overbeek2024termination} contains examples of terminating \pbpostrong{} systems. We have included those examples that have DPO equivalents; see Example~\ref{ex:monic:triangle} and Example~\ref{ex:monic:simple}.} under different combinations of the arctic, tropical and arithmetic semiring.\footnote{The benchmarks were obtained on a laptop with an i7-8550U cpu having 4 cores, a base clock of 1.8GHz and boost of 4GHz.}
The techniques proposed in the current paper successfully establish termination of all examples, with the exception of~\cite[Example 6]{plump2018modular}. 
See Example~\ref{ex:limitations} for a discussion about this system.

\section{Related Work}
\label{sec:related}
\label{discussion:termination:of:graph:transformation}

The approaches~\cite{bruggink2014,bruggink2015} by Bruggink et al.\ are the most relevant to our paper. Both employ weighted type graphs for multigraphs to prove termination of graph transformation.
We significantly strengthen the weighted type graphs technique for rewriting with monic matching and we generalize the techniques to arbitrary categories and different variants of DPO.

Plump has proposed two systematic termination criteria~\cite{plump1995on,plump2018modular} for hypergraph rewriting with DPO.
The approach in~\cite{plump1995on} is based on the concept of forward closures. 
The paper~\cite{plump2018modular} introduces a technique for modular termination proofs based on sequential critical pairs. The modular approach can often simplify the termination arguments significantly. However, the smaller decomposed systems still require termination arguments.
Therefore, our techniques and the approach from~\cite{plump2018modular} complement each other perfectly.
Our technique can handle all examples from the papers~\cite{plump1995on,plump2018modular}, except for one (see Example~\ref{ex:limitations} for discussion of this system).

In~\cite{overbeek2024termination}, we have proposed a termination approach based on weighted element counting.
Like the method presented in this paper, the method~\cite{overbeek2024termination} is defined for general categories (more specifically, rm-adhesive quasitoposes), but for \pbpostrong{}~\cite{overbeek2023quasitoposes,overbeek2024phdthesis} instead of DPO. \pbpostrong{} subsumes DPO in the quasitopos setting~\cite[Theorem 72]{overbeek2023quasitoposes} if both the left-hand morphism $l$ of DPO rules and matches $m$ are required to be regular monic.
The weighted element counting approach defines the weight of an object $G$ by means of counting weighted elements of the form $t : T \to G$, meaning it is roughly dual to the approach presented in this paper. 

An example that the weighted type graph approach can prove, but the weighted element counting approach cannot prove, is Example~\ref{example:tree:counter}~\cite[Example 6]{bruggink2015}, which requires measuring a global property rather than a local one. 
Conversely, Example~\ref{ex:limitations}~\cite[Example 6]{plump2018modular} is an example that seems unprovable using weighted type graphs, while the element counting approach can prove it rather straightforwardly~\cite[Example~5.9]{overbeek2024termination}. So it appears the two methods may be complementary even in settings where they are both defined.

Levendovszky et al.~\cite{injective:2007} propose a criterion for termination of DPO rules with respect to injective matching. Roughly speaking, they show that a system is terminating if the size of repeated sequential compositions of rules tends to infinity. It remains to be seen if this criterion can be fruitfully automated.

Bottoni et al.~\cite{hlr2005termination} introduce a framework for proving termination of high-level replacement units (these are systems with external control expressions). They show that their framework can be used for node and edge counting arguments which are subsumed by weighted type graphs (see further~\cite{bruggink2015}).

There are also various techniques that generalize TRS methods to term graphs~\cite{plump1999term,plump97simplification,moser2016kruskal} and drags~\cite{dershowitz2018graph}.

\begin{remark}
The current paper is an extension of~\cite{gwtg2024icgt} with full proofs, additional examples, and a gentle introduction to our technique (in Section~\ref{sec:overview}). Additionally, we repair a mistake in Definition~\ref{def:traceability} concerning the definition of strong traceability.
\end{remark}

\section{Future Work}
\label{sec:future}

There are several interesting directions for future research:
\begin{enumerate}
\item \label{future:frameworks}
    Can the weighted type graphs technique be extended to other algebraic graph transformation frameworks, such as SPO~\cite{lowe1993algebraic}, SqPO~\cite{corradini2006sesqui}, PBPO~\cite{corradini2019pbpo}, AGREE~\cite{corradini2020algebraic} and \pbpostrong{}~\cite{overbeek2023quasitoposes}?
\item 
    Can the technique be enhanced to deal with examples like Example~\ref{ex:limitations}?
    (See also Open Question~\ref{open:epi}.)
\item
    Can the technique be extended to deal with negative application conditions~\cite{dpo:negative}?
\item 
    Are there semirings other than arctic, tropical and arithmetic that can be fruitfully employed for termination proofs with the weighted type graph technique?
\end{enumerate}
Concerning~\eqref{future:frameworks}, an extension to \pbpostrong{} would be particularly interesting.
In~\cite{overbeek2023quasitoposes}, it has been shown that, in the setting of quasitoposes, every rule of the other formalisms can be straightforwardly encoded as a \pbpostrong{} rule that generates the same rewrite relation.\footnote{The subsumption of SPO by \pbpostrong{} in quasitoposes remains a conjecture~\cite[Remark 23]{overbeek2023quasitoposes}, but has been established for \Graph{}.}
However, our proofs (Lemma~\ref{lem:one:side}) crucially depend on a property specific to pushouts (Lemma~\ref{lem:bijection}). It will therefore be challenging to lift the reasoning to the pullbacks in \pbpostrong{} without (heavily) restricting the framework (to those pullbacks that are pushouts).

\bibliographystyle{alphaurl}
\bibliography{main}

\end{document}